\documentclass[10pt,conference]{IEEEtran}

\usepackage{hyperref}
\AtBeginDocument{%
  }

\usepackage{graphicx} % Required for inserting images

\AtBeginDocument{%
  }

\usepackage{xspace}
\usepackage{color}
\usepackage{listings}
\usepackage{subcaption}
\usepackage{amsmath}
\usepackage{amsthm}
\usepackage{amssymb}
\newtheorem{lemma}{Lemma}[section]

\definecolor{listinggray}{gray}{0.9}
\definecolor{lbcolor}{rgb}{0.9,0.9,0.9}
\newtheorem{theorem}{Theorem}

\newcommand{\weiyu}[1]{{\color{purple}\bf{#1}}}
\newcommand{\mycomment}[1]{}

\newcommand{\Tool}{\textsc{PMRobust}\xspace}
\newcommand{\tool}{\textsc{PMRobust}\xspace}
\newcommand{\captured}{\textit{captured}\xspace}
\newcommand{\escaped}{\textit{escaped}\xspace}
\newcommand{\clean}{\textit{clean}\xspace}
\newcommand{\clwb}{\textit{clwb}\xspace}
\newcommand{\dirty}{\textit{dirty}\xspace}
\newcommand{\ie}{\hbox{\emph{i.e.},}\xspace}
\newcommand{\eg}{\hbox{\emph{e.g.},}\xspace}

\newcommand{\code}[1]{{\footnotesize\texttt{#1}}}
\newcommand{\pstate}[0]{\ensuremath{p}\xspace}

\newcommand{\Pstate}[0]{\ensuremath{\mathcal{P}}\xspace}

\newcommand{\PMap}[1]{\ensuremath{\mathcal{G}_{\Pstate_{#1}}}\xspace}
\newcommand{\aloc}[0]{\ensuremath{l}\xspace}
\newcommand{\ALoc}[0]{\ensuremath{\mathcal{L}}\xspace}

\newcommand{\estate}[0]{\ensuremath{e}\xspace}
\newcommand{\Estate}[0]{\ensuremath{\mathcal{E}}\xspace}
\newcommand{\EMap}[0]{\ensuremath{\mathcal{G}_{\Estate}}\xspace}

\newcommand{\Astate}[0]{\ensuremath{\mathcal{A}}\xspace}
\newcommand{\AliasMap}[0]{\ensuremath{\mathcal{G}_{\Astate}}\xspace}

\newcommand{\aref}[0]{\ensuremath{r}}
\newcommand{\Refs}[0]{\ensuremath{R}}
\newcommand{\mloc}[0]{\ensuremath{m}}
\newcommand{\Mlocs}[0]{\ensuremath{M}}
\newcommand{\tuple}[1]{\ensuremath \langle #1 \rangle}

\newcounter{bugnumbers}

\newcommand{\overhead}[0]{0.26\%\xspace}

\hypersetup{final}
\begin{document}

\title{Automated Insertion of Flushes and Fences for Persistency}

\author{
    \IEEEauthorblockN{Yutong Guo}
    \IEEEauthorblockA{University of California, Irvine
    \\yutong4@uci.edu}\and
    \IEEEauthorblockN{Weiyu Luo}
    \IEEEauthorblockA{University of California, Irvine
    \\weiyu7@uci.edu}\and
    \IEEEauthorblockN{Brian Demsky}
    \IEEEauthorblockA{University of California, Irvine
    \\bdemsky@uci.edu}\and
}

\maketitle
\begin{abstract}
CXL shared memory and persistent memory allow the contents of memory
to persist beyond crashes.  Stores to persistent or CXL memory are typically
not immediately made persistent; developers must manually flush the
corresponding cache lines to force the data to be written to the
underlying storage. Correctly using flush and fence operations is known to 
be challenging.  While state-of-the-art
tools can find missing flush instructions, they often require
bug-revealing test cases.  No existing tools can ensure the absence of
missing flush bugs.

In this paper, we present \tool, a compiler that automatically inserts
flush and fence operations to ensure that code using persistent memory
is free from missing flush and fence bugs. \Tool employs a novel
static analysis with optimizations that target newly allocated objects.
We have evaluated \tool on persistent memory libraries and
several persistent memory data structures and measured a geometric
mean overhead of \overhead relative to the original benchmarks with
hand-placed flush and fence operations.

\end{abstract}
\section{Introduction\label{sec:intro}}

In several different contexts, hardware designers have developed
systems in which the contents of memory can survive system crashes.
Compute Express Link (CXL) is a new open standard that enables cache
coherent shared memory across a network.  CXL also supports 
memory interface to persistent storage~\cite{samsungcxlssd}.
Finally, the contents of battery-backed NVDIMM's can survive power
outages.

In all of these contexts, the in-memory data manipulated by a
computation can survive machine crashes, and it can be desirable to
ensure that key data structures are crash consistent so that they can
be safely accessed after crashes.  Achieving crash consistency is
complicated by the volatility of CPU caches, necessitating explicit
flush and fence instructions to persist data.~\footnote{CXL can 
optionally use energy storage to implement flush on failure, though
this requires support by both the device and
system components. This is not a panacea; many failure modes
would still result in lost cache lines.}  Software developers must use special flush and fence instructions to force data to be
written back to the underlying memory.

There is a body of work~\cite{psan,yashme,pmtest} on finding
bugs in persistent memory programs that utilize flushes and fences.  These tools range from model checkers~\cite{yat,jaaru}
to various dynamic bug finding
tools~\cite{xfdetector,pmemcheck,agamotto}.  The prevalence of such bugs,
with 183 new bugs reported by various tools~\cite{jaaru, psan,
  witcher-sosp21, pmdebugger-asplos21, xfdetector, pmtest, agamotto}
in a small expert-written program set, underscores the difficulty of
manual management and motivates automating flush/fence insertion via
compilers.

We developed a new tool that automatically generates the necessary
flush and fence operations.  Strict
persistency~\cite{memorypelley2014} ensures that the ``persistency
memory order is identical to volatile memory order''.  The observation
here is that bugs caused by missing flush and fence operations can be
eliminated by making stores become persistent in the same order that
they become visible to other threads.

Enforcing strict
persistency is expensive as it requires threads to stall frequently to wait for previous stores to be persisted. Fortunately, we can further loosen the requirements for strict persistency.  If it
is impossible for the program to observe that stores were persisted in
a different order, we can permit those stores to be reordered, and partly recover from the performance penalties of strict persistency.
Robustness leverages this additional degree of freedom and ensures
that any execution of a program under a weak persistency model is
equivalent to some execution of the program under the strict
persistency model~\cite{psan}.  This suffices to ensure the correct
usage of flush and fence operations as additional flush and fence
operations will not alter the set of possible post-crash program
executions. As missing flushes and fences are a major source of bugs in persistent memory program, a tool that eliminates these concerns greatly simplifies persistent memory programming. 

Robustness to persistency models has been explored before in the context of bug finding.  PSan uses a dynamic analysis combined with random execution or model checking to check robustness.  PSan
suffers from the same limitations as all dynamic analysis---it
requires test cases and may miss bugs that are not revealed by the
test cases.  Thus, the analysis used by PSan cannot be the foundation
of a compiler that automatically inserts flush and fence operations.

\textbf{\tool is
  applicable to CXL shared memory and all non-volatile memory 
  types.  We collectively refer to all of these memory types as
  persistent memory (PM) throughout the paper for convenience.}

This paper makes the following contributions:
\begin{itemize}
\item {\bf Static Analysis:}  It presents the first static analysis that can efficiently analyze full programs for missing flushes or fences using robustness as a correctness criteria.

\item {\bf PM Analysis:} 
It presents the first static analysis that can determine which operations will only modify volatile or local memory, reducing analysis and annotation overhead.   

\item {\bf Automated Flush Insertion:} It presents the first static tool that can both repair missing flush and fence bugs as well as insert any missing fences, freeing the developer from this task.

\item {\bf Eliminates Missing Flush Bugs:} It presents the first tool that ensures the absence of missing flush bugs.

\item {\bf Evaluation:} It evaluates \tool on a wide range benchmarks, incurring minimal runtime overhead.
\end{itemize}

We have made \tool's source code, benchmarks, and scripts for reproducing evaluation results available at: \url{https://github.com/uci-plrg/PMRobust-docker}
\section{Underlying Hardware and Failure Model}
Persistency semantics are important for a range of new memory
hardware. A well-studied example is Px86~\cite{intelx86}, the semantics of persistent memory for x86 processors, which adds a global \emph{persistent buffer} that holds stores after they exited the thread-local store buffers. Stores in the persistent buffer are written back to memory in arbitrary order at cache line granularity, and those not written back are lost upon a crash. To avoid inconsistent state caused by crashes under persistency semantics, fence and flushes instructions are often needed to impose persistency ordering between stores. 

While the problem has received much attention for
persistent memory, the same problem exists for CXL shared
 memory.  Figure~\ref{fig:memorysystem} presents a
graphical overview of CXL shared disaggregated memory.  In this
setting, processing nodes and memory nodes are connected via CXL
networking.  Memory can be coherently shared between multiple hosts
--- memory nodes contain directories that track which cache lines have
been requested by processing nodes.  CXL on x86 is intended to provide
TSO ordering guarantees across machines.

Failures of compute nodes pose a consistency problem.  Compute nodes
cache their writes to CXL shared memory and the contents of these
caches can be written back in an arbitrary order.  The CXL standard
provides for global persistent flush (GPF).  When GPF is supported, it uses
energy reserves to write the contents of the CPU cache back to the underlying storage upon a crash or system shutdown. However, GPF does not cover a wide
range of failure modes, \eg failed cables, failed CXL transceivers,
failed CPUs, failed motherboards, etc, and therefore does not provide a complete solution to the consistency problem.

If even a single machine that is updating a given CXL memory region
fails, the data in that machine's cache is corrupted for all machines
that access the CXL memory region.  In our email discussions with
Intel engineers, it is not yet decided how CXL will expose such
failures to the software layer.  One option is to report those
cache lines as poisoned and throw exceptions on access and another
is to allow software to access the copy of the data that
resides in the CXL memory node and let software manage flushing data
with flush and fence operations in a manner similar to 
PM systems.

\begin{figure}[!htbp]
\begin{center}
  \includegraphics[scale=0.26]{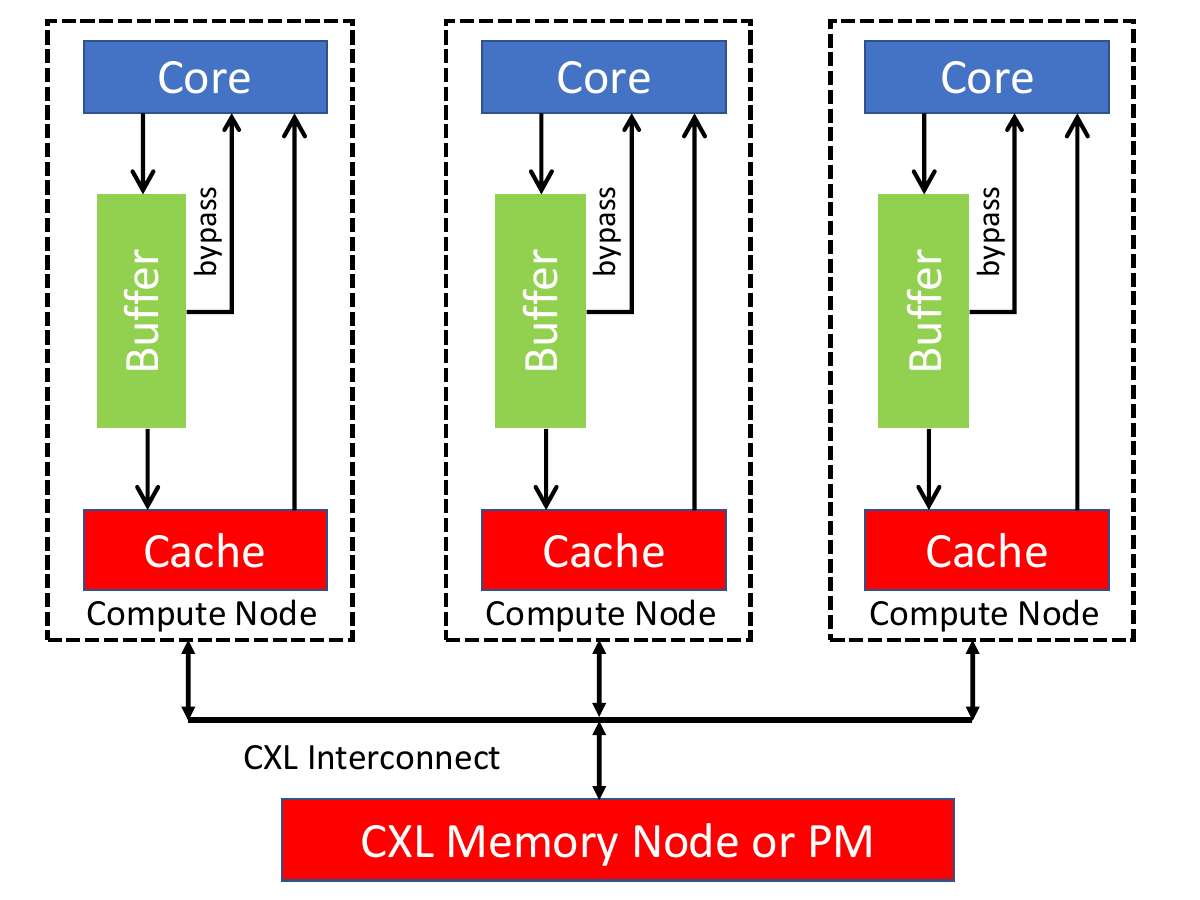}
\end{center}
\vspace{-.2cm}
  \caption{CXL Memory System\label{fig:memorysystem}}
\end{figure}

\textbf{PM can be viewed as a special case of the CXL shared memory
model with the differences that: (1) there is only one processing node and
 (2) after a crash, the node can return having lost the state of its
local cache.}

The x86 architecture provides the following instructions to force the cache
to write data back to persistent storage: (1) the  \code{clflush} and \code{clflushopt} instructions that
flush (and evict) a cache line and (2) the 
\code{clwb} instruction that writes back a cache line potentially without eviction.  Each of these instructions takes as input the address to flush.  The \code{clflush} instruction flushes a cache line immediately while the \code{clflushopt} and \code{clwb} instructions are not guaranteed to flush or write back a cache line until the thread executes a fence instruction.

\section{\Tool}

This section discusses the key ideas behind our approach to
inserting flush and fence operations.

\subsection{Robustness\label{sec:discipline}}

\begin{figure}[!h]
\begin{subfigure}{0.24\textwidth}
{\scriptsize
  \begin{lstlisting}[xleftmargin=1.65cm]
    x = 1;
    y = 1;
  \end{lstlisting}
\caption{\label{fig:precrash}Pre-crash execution}
  }
\end{subfigure}
\begin{subfigure}{0.24\textwidth}
  {\scriptsize
  \begin{lstlisting}[xleftmargin=1.6cm]
    r1 = x;
    r2 = y;
  \end{lstlisting}
  }
  \caption{\label{fig:postcrash}Post-crash execution}
  \end{subfigure}
\caption{\label{fig:robustness}Assume $\code{x = y = 0}$ initially.  If the post-crash execution observes $\code{r2 = 1}$, strict persistency requires that $\code{r1 = 1}$.}
\vspace{-.3cm}
\end{figure}

Figure~\ref{fig:robustness} presents an example that illustrates the
requirements of strict persistency.  Strict persistency requires that
the persistency order for stores respects the happens-before relation.
This means that the store $\code{x = 1}$ must be persisted before the store $\code{y = 1}$, forbidding an execution in which $\code{r1 = 0}$ and $\code{r2 = 1}$.

It may initially appear that robustness would require placing flushes
after every memory access to PM.  However, robustness
only requires that the persistency order for stores respects the
happens-before relation when an execution can potentially observe a
violation of strict persistency.  For example if post-crash executions 
only read from $\code{y}$, then the program is robust even if 
$\code{y = 1}$ is made persistent while $\code{x = 1}$ is not.  This
observation is most relevant for newly created persistent objects
that have not yet become reachable from persistent
data structures.  Thus, it suffices to delay flushing stores to newly created
persistent objects until right before they are inserted into
persistent data structures.

We next present a sufficient set of requirements on flush and fence operations
to ensure robustness.  Figure~\ref{fig:discipline} presents a finite
state machine that captures how to check robustness for the x86-TSO persistency model.  We refer to a
state in this finite state machine
as an \textit{escape persistency state}.  The
finite state machine captures the set of legal transitions for cache
lines through escape persistency states.
If there is an object that is both escaped and non-clean and the program writes to an escaped object, then this is a \textit{robustness violation}.
A key insight is differentiating
between (1) memory locations that are \textit{captured} by the local
thread and thus stores to the memory location would not be visible if
the program crashed and (2) memory locations that
have \textit{escaped} to become reachable from the roots of persistent data structures, and thus stores would be visible if the program
crashed. With this distinction, captured objects do need to flushed according to the volatile order, because they cannot be read from in a post-crash execution.

\begin{figure}[!htbp]
\begin{center}
\includegraphics[scale=0.31]{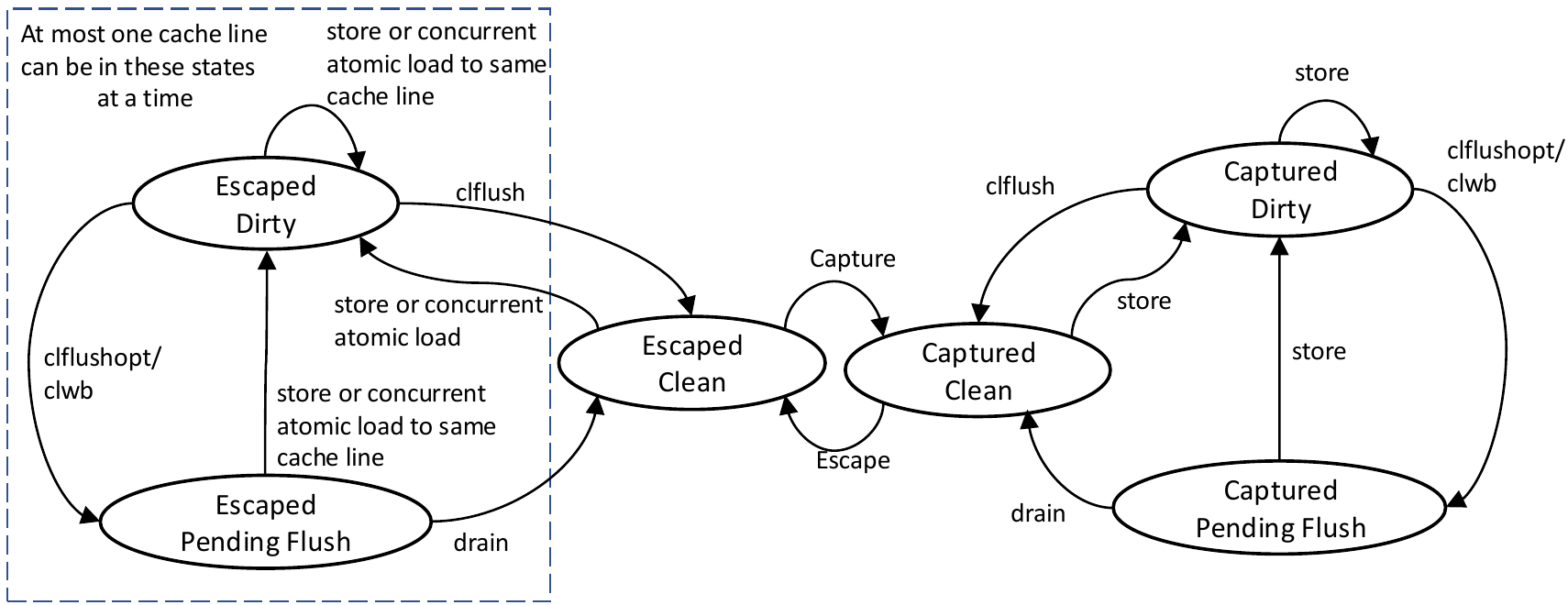}
\end{center}
\caption{Using Flush \& Drain Operations to Ensure Robustness\label{fig:discipline}}
\end{figure}

\noindent \textbf{Captured Objects:} Memory locations are captured
when there exists no path from the persistent data structure roots to
the memory location.  Stores to captured memory locations are not
visible after a crash, and thus it is safe to delay flushes
until immediately before the memory location escapes via insertion
into a persistent data structure.  This can have several
benefits---first, it becomes possible to use optimized flushes
like $\code{clflushopt}$ or $\code{clwb}$ on several
cache lines and amortize the cost of the fence operation across
multiple flushes.  Second, it is sufficient to handle
multiple stores to the same cache line with a single
flush.

\noindent \textbf{Escaped Objects:} Memory locations have escaped if
they are reachable from a persistent data structure.  Escaped memory
locations require a more expensive flush insertion approach.  If
consecutive stores happen to be to the same cache line, their
persistency order is enforced automatically due to cache coherence.
If they are to different cache lines, it is necessary to flush the
first store before performing the second store.

\begin{figure}[h]
\begin{subfigure}{0.24\textwidth}
{\scriptsize
\hspace{1.25cm} Thread 1:
\begin{lstlisting}[xleftmargin=1.7cm]
  x = 1;
  clflush(&x);
\end{lstlisting}  
\hspace{1.25cm} Thread 2:
\begin{lstlisting}[firstnumber=last, xleftmargin=1.7cm]
  r1 = x;
  y = r1;
\end{lstlisting}
\caption{\label{fig:rrprecrash}Pre-crash execution}
  }
\end{subfigure}
\begin{subfigure}{0.24\textwidth}
  {\scriptsize
  \begin{lstlisting}[xleftmargin=1.6cm]
    r2 = x;
    r3 = y;
  \end{lstlisting}
  }
  \caption{\label{fig:rrpostcrash}Post-crash execution}
  \end{subfigure}
\caption{\label{fig:racyreads}A non-robust program that is missing a flush on load. Assume that $\code{x = y = 0}$
initially and all accesses are atomic, it is possible that $\code{r2 = 0}$ and $\code{r3 = 1}$}
\end{figure}

Atomic loads require extra care to ensure robustness.  Consider the
example in Figure~\ref{fig:racyreads}.  If the pre-crash execution
crashes before the $\code{clflush}$ in Thread 1 completes,
but after the store to $\code{y}$ in Thread 2 has been made
persistent, it is possible for the post-crash execution to observe
$\code{r2 = 0}$ and $\code{r3 = 1}$, violating strict consistency.
Robustness in this case requires a $\code{clflush}$ to the
cache line of $\code{x}$ after Thread 2 reads from $\code{x}$.  This
example also has an implication for all stores to escaped memory
locations inside of critical sections---the store must be persisted
before the lock is released.

\subsection{Analyzing Robustness}

\begin{figure}[!h]
\begin{tabular}{c}
\begin{lstlisting}
struct Node {
  int data;
  struct Node * next;
};

struct Stack {
  struct Node * top;
};

void push(struct Stack *s, int val) { //s: <esc, clean>
  struct Node * head = s->top; /*@ \label{line:head}@*/
  struct Node * n = pmalloc(sizeof(struct Node));
  //n: <cap, <clean, clean>>
  n->data = val; /*@ \label{line:n-store1}@*/
  n->next = head; /*@ \label{line:n-store2}@*/
  //n: <cap, <dirty, dirty>>
  s->top = n; /*@ \label{line:commit-store}@*/
  //s: <esc, dirty>, n: <esc,<dirty, dirty>>
}
\end{lstlisting}
\end{tabular}
\caption{\label{fig:stack}A Persistent Stack. esc = escaped, cap = captured}
\end{figure}

We begin with an example to illustrate key concepts in our
approach to analyzing PM code.  Figure~\ref{fig:stack} presents a
single-threaded persistent stack.  The \code{push} method adds a new
value to the top of the stack.  It calls
\code{pmalloc} to allocate a new stack node, stores the value
\code{val} to the node, and updates the node's \code{next} field.
Then it flushes the new node and updates the top of the stack to
reference the new node.  Finally, it flushes the update to
the top of the stack.

In this example, the stack \code{s} and node \code{n} are persistent
variables and have one of the states in
Figure~\ref{fig:discipline}. We assume \code{s} and \code{n} are on
different cache lines.  The stack \code{s} is the root of the
persistent data structure and is escaped initially.  When the node
\code{n} is created, both of its fields have the state
$\tuple{\captured, \clean}$ initially.  Node \code{n} has the
state $\tuple{\captured, \dirty}$ for both fields after the stores at
lines~\ref{line:n-store1} and~\ref{line:n-store2}.  The commit store
at line~\ref{line:commit-store} makes \code{n} \escaped as \code{n} is
reachable from the persistent data structure root, and the state of
\code{s->top} transitions to $\tuple{\escaped, \dirty}$.  This is a
violation of robustness as both \code{n} and \code{s} are \escaped and
\dirty. Thus, we must insert a flush before
line~\ref{line:commit-store}.
 
\subsubsection{Basic Approach\label{sec:verification}}

We next discuss our approach to analyzing PM programs.  Our approach
builds on the finite state formulation for ensuring robustness from
Section~\ref{sec:discipline}.  Our analysis is structured as a
standard fixed-point dataflow analysis.  The basic idea is to use a
static analysis to compute at each program point a mapping from memory
locations to a set of potential escape persistency states.  The
transfer function implements the finite state machine from
Figure~\ref{fig:discipline}.  We apply this machine to each of the
potential escape persistency states for a given memory location to
generate a new output set of escape persistency states.

The analysis checks several correctness properties.  The first is that
the program does not take a forbidden transition that would violate
robustness such as having
multiple escaped and non-clean cache lines.  We compute summaries of
the effect of method calls by extending escape persistency states to
tracks the corresponding initial persistency states when the method
was first called.  When the analysis of the method is complete, the
static analysis has determined how the method changes each of the
possible escape persistency states.

Loads pose a challenge because they allow another thread
to observe a store before it is made persistent, and that thread may
later store a value that was derived from the value returned by the
load.  The later store can potentially be persisted before the initial
store, and thus a crash can leave the PM in an
inconsistent state.  For example, in Figure~\ref{fig:racyreads}, it is
possible for \code{y = r1} to be persisted before \code{x = 1}, and
then the post-crash execution would read \code{r2 = 0} and \code{r3 =
  1}.

This problem can be solved by inserting a flush 
immediately after every load, but this incurs an overhead.  We
consider two cases for loads:
\begin{enumerate}

\item {\bf Non-atomic Loads:} Here we assume programs are data race
  free, as they otherwise have undefined semantics in languages such
  as C/C++.  In the case of a non-atomic load, we require that
  non-atomic stores be persisted before any release operation such as
  an unlock operation.  Then there is no issue because store is
  persisted before its mutex is released and thus before it can be
  read.

\item {\bf Atomic Loads:} Atomic operations allow multiple threads to access memory
  without acquiring a lock.  As a result, we cannot assume the flush instruction after the corresponding store has completed before the atomic load.
  For atomic loads, we use
  FliT~\cite{flit} to ensure the corresponding store is flushed.  In
  the analysis, we address this issue by having atomic loads change
  the persistency state to dirty.  This forces the thread to flush the
  data before performing other visible stores.
  \end{enumerate}

\subsection{Detecting References to Persistent Memory}

\Tool uses a CFL-reachability-based alias analysis introduced by Zheng
and Rugina~\cite{cflalias} to distinguish PM locations from
non-persistent locations. First, the analysis requires a set of
user-configured PM allocators, and the pointers
returned by these allocators are identified as the initial set of PM
pointers. The aliases of known PM pointers are iteratively computed
and added to the set until a fixed point is reached. The
CFL-reachability-based formulation of the aliasing relation is precise
and enables a demand-driven algorithm, which finds all pointers that may
reach persistent memory. This helps \tool avoid analyzing
a large number of volatile memory pointers.

Our analysis is adapted from an existing
implementation from the LLVM-8 codebase~\cite{llvm-8} that computes
aliases between all pairs of pointers. We modified the analysis to
only explore aliases of PM pointers following the original
demand-driven formulation~\cite{cflalias}. We also added type-based
field-sensitivity to the analysis to track whether each offset of a
struct type is a PM pointer, and treat all objects of the same type
uniformly. This approach intentionally sacrifices some precision compared to
full field sensitivity because PM programs likely use
specialized data types for PM as we observe in the PM benchmarks.

\section{Intraprocedural Analysis\label{sec:intra}}

In this section, we first discuss the core intraprocedural analysis.
Later, in Section~\ref{sec:inter} we will extend this analysis to the
interprocedural context and to handle arrays.  The intraprocedural
part is implemented as a standard forward dataflow analysis.  The
algorithm maintains a program state at each instruction.

\subsection{Preliminaries\label{sec:prelim}}

The instructions that we analyze are atomic and non-atomic loads and stores, atomic RMWs,
assignments, and flushes and fences.  An object can occupy
multiple cache lines and thus an object reference $\aref \in \Refs$
can be used depending on the field to access one of several different
cache lines. Objects are by default not aligned to cache lines, and
thus the static analysis may not know whether two
different fields reside on the same cache line.  Thus, we model a
memory location $\mloc \in \Mlocs$ as the combination $\aloc =
\tuple{\aref, n} \in \ALoc$ of a reference (variable that references a memory location) $\aref \in \Refs$ and a
non-negative offset $n \in \mathbb{Z}^{0+}$ from that reference. 

We next describe our core analysis approach.  For each PM location,
our analysis must compute: (1) whether a reference to that memory
location may have escaped to PM, and (2) whether all
stores to the memory location have been flushed to PM.
Although the original finite state machine in
Figure~\ref{fig:discipline} combines both properties into a single
finite state machine, we separated the two properties into two finite
state machines to simplify the presentation.

Figure~\ref{fig:escfsm} presents a finite state machine that
captures whether a memory location has escaped at a given program
point, and Figure~\ref{fig:esclattice} presents a lattice for our
escape analysis.  We use an escape state $\estate \in \Estate$ to
represent one of the two escape values, $\{\captured, \escaped\}$,
from the escape analysis lattice. We use a may-escape analysis to
conservatively catch all cases where a reference might have escaped.
Thus, we have the \escaped value lower in the lattice, and a merge of
an \escaped value with a \captured value yields the \escaped
state.  The analysis computes a map $\EMap \subseteq \Refs \times
\Estate$ from memory locations to escape states at each point in the
program.  The meet operator $\sqcap: \Estate \times \Estate
\rightarrow \Estate$ is defined by $\estate_1 \sqcap \estate_2 =
\text{lower}(\estate_1, \estate_2)$, which returns the lower of the
two lattice values.  We write $\estate_1 \succeq \estate_2$ if
$\estate_1$ is higher than or equal to $\estate_2$ in the lattice.

\begin{figure}[!htbp]
\begin{subfigure}{.24\textwidth}
\vspace{-.7cm}
\begin{center}
\includegraphics[scale=0.35]{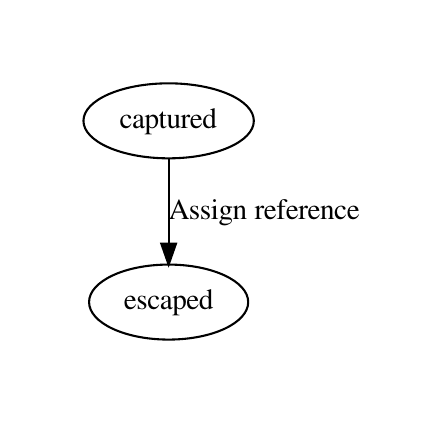}
\end{center}
\vspace{-.6cm}
\caption{Finite State Machine\label{fig:escfsm}}
\end{subfigure}
\begin{subfigure}{.24\textwidth}
\vspace{-.7cm}
\begin{center}
\includegraphics[scale=0.35]{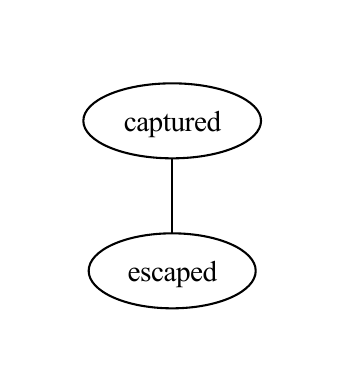}
\end{center}
\vspace{-.6cm}
\caption{Lattice\label{fig:esclattice}}
\end{subfigure}
\caption{Lattice and FSM for Escape Analysis}
\end{figure}

Figure~\ref{fig:clfsm} presents a finite state machine that summarizes
the semantics for persistency state of a memory location, and
Figure~\ref{fig:cllattice} presents the corresponding lattice.
We use a persistency state $\pstate \in \Pstate$ to
represent one of the three persistency state values, $\{ \clean,
\clwb, \dirty \}$.  The lattice is ordered in this fashion, because we
need to know whether a memory location may require a fence 
(\clwb) or whether it may require a fence and flush
(\dirty).  For example, if a reference is \clean on one path to a node
and \clwb on a different path to the node, the analysis must
conservative assume it is \clwb at the merge point.

The core analysis computes a map $\PMap{} \subseteq \ALoc \times
\Pstate$ from memory locations to persistency states.  The meet
operator $\sqcap: \Pstate \times \Pstate \rightarrow \Pstate$ and the
ordering operator $\succeq$ for $\Pstate$ are defined similar to the
ones for $\Estate$.

\begin{figure}[!htbp]
\begin{subfigure}{.24\textwidth}
\vspace{-.2cm}
\begin{center}
\includegraphics[scale=0.35]{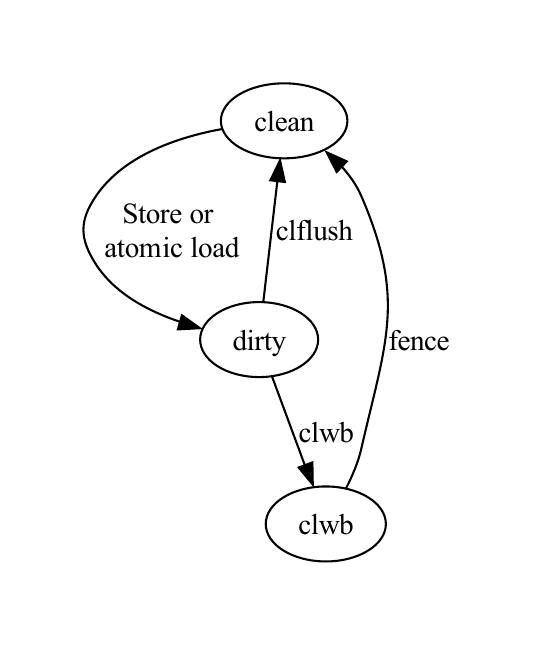}
\end{center}
\vspace{-.5cm}
\caption{Finite State Machine\label{fig:clfsm}}
\end{subfigure}
\begin{subfigure}{.24\textwidth}
\vspace{-.2cm}
\begin{center}
\includegraphics[scale=0.35]{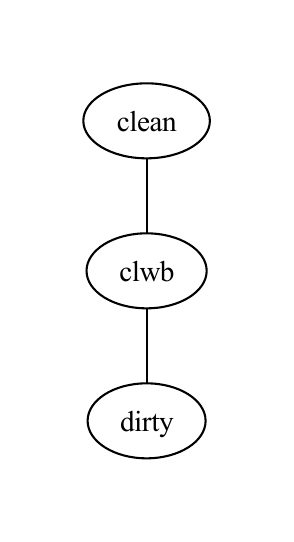}
\end{center}
\vspace{-.5cm}
\caption{Lattice\label{fig:cllattice}}
\end{subfigure}
\caption{Lattice and FSM for Persistency State Analysis}
\end{figure}

\newcommand{\tspace}{\vspace{.01cm}}
\newcommand{\ispace}{\vspace{.1cm}}

\begin{figure}[!htb]
\scriptsize
\begin{tabular}{p{0.2\linewidth} | p{0.65\linewidth}}
  \code{statement} & {
  $\EMap' = (\EMap - \text{KILL} ) \cup \text{GEN}$
  }\ispace\\
  \hline \hline
  \tspace\code{y=x} & {\tspace$\begin{aligned}
  U =& \AliasMap(\code{x}) \cup \{ \code{y} \}\\
  \AliasMap' =& \{ \tuple{\aref, S} \mid \tuple{\aref,S} \in \AliasMap \wedge \aref \notin U \} \cup \\
  & \{ \tuple{\aref, U} \mid \aref \in U \} \\
  \text{GEN} =&\{ \tuple{\code{y}, \estate{}} \mid \tuple{\code{x}, \estate} \in \EMap{}\} \\
  \text{KILL} =&\{ \tuple{\code{y}, *} \mid \tuple{\code{x}, \estate} \in \EMap{}\} \\
  \end{aligned}$\ispace}\\
  \hline
  
  \tspace\code{*y=x} \linebreak or \linebreak \code{*y=\&x->f} & {\tspace$\begin{aligned}
  U =& \AliasMap(\code{x}) \\
  \text{GEN} =&\{ \tuple{\code{y}, \escaped} \mid \tuple{\code{x}, \estate} \in \EMap{} \wedge \aref \in U\} \\
  \text{KILL} =&\{ \tuple{\code{y}, *} \mid \tuple{\code{x}, \estate} \in \EMap{} \wedge \aref \in U\} \\
  \end{aligned}$\ispace}\\
  \hline

  \tspace\code{y=*x} & {\tspace$\begin{aligned}
  U =& \AliasMap(\code{y}) \setminus  \{ \code{y} \}\\
  \AliasMap' =& \{ \tuple{\aref, S} \mid \\
  & \tuple{\aref,S} \in \AliasMap \wedge r \notin \AliasMap(\code{y}) \} \cup \\
  & \{ \tuple{\code{y}, \{ \code{y} \}} \} \cup \{ \tuple{\aref, U} \mid \aref \in U \} \\
  \text{GEN} =& \{ \tuple{\code{y}, \escaped} \} \\
  \text{KILL} =& \{ \tuple{\code{y}, *} \}
  %\hspace{-1.6cm}\text{KILL} =& \{ \tuple{\code{y}, \estate} \mid \tuple{\code{y}, \estate} \in \EMap{} \}
  \end{aligned}$\ispace}\\

  \hline
\end{tabular}
\caption{Transfer Functions for Escape Analysis, where \code{x} and \code{y} point to PM locations\label{fig:esctransfer}}
\end{figure}
\begin{figure}[!htb]
\scriptsize
\begin{tabular}{p{0.2\linewidth} | p{0.65\linewidth}}
  \code{statement} & {$\begin{aligned} \PMap{}' = (\PMap{} - \text{KILL} ) \cup \text{GEN} \end{aligned}$\ispace}\\
  \hline \hline

  \tspace\code{x->f=v} & {\tspace$\begin{aligned}\hspace{-2.1cm}\text{GEN}=&\{\tuple{\tuple{\code{x}, \text{offset}(\code{f})},\dirty}\}\\
  \text{KILL} =& \{ \tuple{\tuple{\code{x}, \text{offset}(\code{f})},\text{*}}\}\end{aligned}$\ispace}\\
  \hline

  \tspace\code{y=x->f} \text{where \code{f} is an} \text{atomic field} & {\tspace$\begin{aligned}\hspace{-2.2cm}\text{GEN}=&\{\tuple{\tuple{\code{x}, \text{offset}(\code{f})},\dirty}\}\\
  \text{KILL} =& \{ \tuple{\tuple{\code{x}, \text{offset}(\code{f})},\text{*}}\}\end{aligned}$\ispace}\\
  \hline

  \tspace\code{flush(\&x->f)} & {\tspace$\begin{aligned}\hspace{-2.2cm}\text{GEN}=&\{\tuple{\tuple{\code{x}, \text{offset}(\code{f})},\clean}\}\\
  \text{KILL} =& \{ \tuple{\tuple{\code{x}, \text{offset}(\code{f})},\text{*}}\}\end{aligned}$\ispace}\\
  \hline

  \tspace\code{clwb(\&x->f)} & {\tspace$\begin{aligned}
  \text{GEN}=&\{ \tuple{\aloc, \clwb} \mid
  \tuple{\aloc, \dirty} \in \PMap{} \wedge \aloc = \tuple{\code{x}, \text{offset}(\code{f})}\}\\
  \text{KILL}=&\{ \tuple{\aloc, *} \mid
  \tuple{\aloc, \dirty} \in \PMap{} \wedge \aloc = \tuple{\code{x}, \text{offset}(\code{f})}\}\\
  \end{aligned}$\ispace}\\
  \hline

  \tspace\code{fence} & {\tspace$\begin{aligned}
  \text{GEN}=&\{ \tuple{\aloc, \clean} \mid \tuple{\aloc, \clwb} \in \PMap{} \}\\
  \text{KILL}=&\{ \tuple{\aloc, *} \mid \tuple{\aloc, \clwb} \in \PMap{} \}\\
  \end{aligned}$\ispace}\\
  \hline

  \end{tabular}
\caption{Transfer Functions for Persistency State Analysis\label{fig:perstransfer}}
\end{figure}

\subsection{Checking Whether Objects are Captured\label{sec:esctransfer}}

We take a simple approach to escape analysis
--- once a reference to a newly allocated struct or array is stored to
any place other than a variable, we assume it has escaped.  The key
ideas of the analysis are that a newly allocated object starts in the
captured state.  For example, the statement \code{x=new} would result
in the analysis computing that \code{x} is in the captured state at
the program point immediately after this statement.  This is stored in
the map $\EMap$.  The analysis then computes the sets of variables
that may reference the same object.  After the statement \code{x=new},
the analysis would compute that \code{x} is the only variable to
reference the memory it references.  The analysis stores this
information in the alias map $\AliasMap \subseteq \Refs \times
\mathbb{P}(\Refs)$ from references to their aliases, where
$\mathbb{P}(\Refs)$ denotes the power set of $\Refs$.  If the value in
a variable \code{x} is stored to some heap location, the variable
\code{x} and all variables that may reference the same heap location
are marked as $\escaped$.

Figure~\ref{fig:esctransfer} present the transfer functions of our escape analysis.  We use the form $\EMap' = (\EMap -
\text{KILL} ) \cup \text{GEN}$ and use the *
symbol to match arbitrary states. We next discuss the transfer functions for key statements:

\noindent {\bf Assignments:}
When there is an assignment 
\code{y=x}, we replace the alias set of \code{y} (and all other aliases of \code{x}) with the aliases of \code{x} plus \code{y} itself.
At the same time, we update $\EMap$ to assign \code{y} to have the same escape state as \code{x}.

\noindent {\bf Stores:} When a store 
\code{*y = x} or \code{*y = \&x->f} stores to the address \code{x}
or the address of one of its fields \code{\&x->f} to any location, we
consider \code{x} and its aliases as escaped.
As an example, line~\ref{line:commit-store} in
Figure~\ref{fig:stack} makes \code{n} \escaped.

\noindent {\bf Loads:} When a load \code{y =
  *x} reads from \code{x}, and \code{y} points to persistent
memory, the analysis marks \code{y} as \escaped as the loaded value comes from a dereference rather than directly from a variable. Since \code{y} is overwritten, we set the alias set of \code{y} to only include itself, and remove \code{y} from the alias sets of its previous aliases.  Note that since \code{y} is escaped, correct alias information is not necessary.

\subsection{Analyzing Persistency States}
We next describe our persistency state analysis.
Figure~\ref{fig:perstransfer} presents the transfer functions for the
persistency state analysis, where we assume \code{x} is a PM location.
We express transfer
functions using GEN and KILL sets. The reader may note that variables
may alias, but this analysis does not track aliasing information.  The
key observation is that aliasing does not violate soundness---it
simply means that the same variable that is used to perform a store
must be used to flush the value. The lack of must alias information
may result in false positives in cases where one alias is used to
perform a store and another is used to flush the store.  In
Figure~\ref{fig:perstransfer}, we use $\tuple{\code{x},
  \text{offset}(\code{f})}$ to denote the memory location of
\code{x->f}.

\noindent {\bf Store/Atomic Store:} When an atomic or non-atomic store
writes to a field \code{x->f}, the persistency state of \code{x->f}
(\ie $\tuple{\code{x}, \text{offset}(\code{f})}$) is updated as \dirty, indicating that the location needs to be flushed.

\noindent {\bf Atomic Load:} As we mentioned in
Section~\ref{sec:verification}, an atomic load changes the state of
the loaded variable or field to \dirty.  The analysis removes the
old persistency state of \code{x->f} and marks it as \dirty.

\noindent {\bf Atomic RMW:} An atomic RMW is a combination of a fence, an atomic
load, and an atomic store, so we apply the transfer functions for a fence, atomic load
and store in sequence.
 
An atomic CAS is an atomic RMW if successful and an fence and atomic load,
otherwise.  Since our transfer functions have the same effects when
storing to a field \code{x->f} and when performing an atomic load from
\code{x->f}, we consider atomic CASs the same as atomic
RMWs.

\noindent {\bf Flush:} When flushing the address of a field
\code{x->f} with the stronger \code{clflush}, the
persistency state of \code{x->f} becomes clean.  Flushing the address
of a field \code{x->f} with \code{clwb} or \code{clflushopt}
changes its state to \code{clwb} if it was \dirty,
indicating it will become \clean at the next fence.

\noindent {\bf Fence:} For fences, the transfer function
leaves the persistency states untouched unless they are
\clwb. The locations whose persistency states are \clwb
are changed to the \clean state.

\subsection{Intraprocedural Violation Detection\label{sec:intra-error}}

\Tool consider three categories of violations. The first two
categories are 1) unflushed PM locations at function exits; and 2) a
store to an escaped PM location when a different PM location is
already escaped and non-clean. The third category involving arrays
will be discussed in Sections~\ref{sec:array}.

The first category of violations is checked at function exits. When we
complete the analysis of a function, we get the program states at each
function exit and take a union of the states by using meet
operators. If the state of any PM location that is not a function
parameter or the return value is escaped and non-clean, we report it
as a violation.

The second category of errors is checked at every program point. If a
PM location is escaped and non-clean and we perform a store to an escape PM location,  this is a
violation.  To address violations involving multiple threads, when
there is an escaped and non-clean location, a release operation to a
non-PM location (atomic store release or unlock) is
reported as a violation.  Recall from Section~\ref{sec:prelim}, a PM
location $\aloc$ is a pair $\tuple{\aref, n}$ of a reference and an
offset from the reference, so an escaped PM object with two or more
fields being non-clean is also reported as a violation.

\section{Durability}
In addition to robustness violations, the first category of violation from the previous section also includes all durability bugs, i.e.
any escaped PM objects that is not made durable by the end of the program. This follows from a simple argument: 
 the only objects that may be escaped and non-clean are the return values and parameters of the top-level main function, but none of these
objects can be PM objects. 
\section{Interprocedural Analysis\label{sec:inter}}

In this section, we first discuss extending the core analysis to be
interprocedural.  Then we discuss how we detect bugs that involve
objects reachable from function parameters and bugs that involve
multiple functions.  Lastly, we describe our array support
and how we detect references to PM.

\subsection{Context Sensitivity}

We handle function calls in a context sensitive manner using function
summaries. Each function has a \textit{function summary table}
that maps calling contexts to summarized results.

For a function with $n$ parameters, its \textit{calling context} has
the form $C \in (\Estate \times \Pstate)^n$,  containing the abstract
escape and persistency states for each function parameter. To reduce
the possible number of calling contexts, for a function parameter with
$m$ fields, we collapse the $m$ persistency states to a single
abstract persistency state with the abstraction:
\begin{align*}
\textit{Abs}&(\tuple{\estate, \pstate_1, ..., \pstate_m}) = \tuple{\estate, \text{lowest}(\pstate_1, ..., \pstate_n)}
\end{align*}
In other words, we use a single state (the lowest one) for all fields of a given parameter. 

Each calling context is then mapped to a \textit{summarized result} in
the function summary table. For a function $F$ with $n$ parameters, a
summarized result for a calling context $C$ has the form
$R_{\tuple{F,C}} \in (\Estate \times \Pstate)^{n+1}$, containing
abstract states for each function parameter and the return value. The
last element may be ignored for functions with no return values. The
abstraction can be reverted to get back program states by
approximating every field with the same persistency state:
\begin{align*}
\textit{AbsRev}&(\tuple{\estate, \pstate}) = \tuple{\estate, (\pstate, ..., \pstate)}
\end{align*}
This allows summarized results to be used to update the program state
after a function call.

To extend our alias analysis to be interprocedural, we also store the
aliasing information between function parameters and the return value
to summarized results. The summarized result additionally uses a
\textit{markObjDirEsp} bit to record if there are any escaped objects
that become non-clean in the callee, potentially causing violations with
other escaped non-clean objects in the caller, but are not captured by
the summarized result because they become clean again before the
callee returns. These extensions are straightforward and are omitted in
the representation here.

The interprocedural analysis
uses a worklist algorithm operating on pairs $\tuple{F, C}$ of
function $F$ and calling context $C$.  The worklist initially includes
all functions with the calling contexts of all parameters having the
lowest state $\tuple{\captured, \clean}$. When we complete the
analysis of a function with a calling context, we update the function
summary table, and every time $F$'s summarized results are updated, we
push all callers of $F$ with their calling contexts to the worklist.

To speed up the convergence of the algorithm, we allow summarized
results for a function $F$ to be approximated from results for higher
contexts. For two calling context $C = \tuple{\tuple{\estate_1,
    \pstate_1}, ..., \tuple{\estate_n, \pstate_n}}$ and $C' =
\tuple{\tuple{\estate_1', \pstate_1'}, ..., \tuple{\estate_n',
    \pstate_n'}}$, we say that $C$ is higher than $C'$ if $\estate_i
\succeq \estate_i'$ and $\pstate_n \succeq \pstate_n'$ for all $i$.
The meet operator $\sqcap: (\Estate \times \Pstate)^{n+1} \rightarrow
(\Estate \times \Pstate)^{n+1}$ for summarized results is also defined
as a point-wise meet.  There are then three cases when processing a
pair $\tuple{F, C}$:

\begin{itemize}
  \item \textbf{Case 1:} If we have already analyzed $F$ with the
    calling context $C$, Then the summarized result $R_{\tuple{F,C}}$
    is used to approximate the state of the parameters and the return
    value of $F$ after the call site.

  \item \textbf{Case 2:} Otherwise, if we have only analyzed $F$ with
    calling contexts higher that $C$, we can take the summarized
    results $R_{\tuple{F, C'}}$ for all calling contexts $C'$ higher
    than $C$, and use the merged result of this set of summarized
    results via the meet operator as an approximation, and push the
    pair $\tuple{F,C}$ to the worklist. This choice ensures
    monotonicity when processing call sites and thus preserves the
    termination guarantee for dataflow analysis.

  \item \textbf{Case 3:} If we have not analyzed $F$ with $C$ or any calling context higher than $C$ before, we approximate all parameters and the return value of $F$ as having the state $\tuple{\captured, \clean}$ and push the pair $\tuple{F,C}$ to the worklist.
\end{itemize}

\iffalse
\begin{figure}[!h]
\begin{subfigure}{0.4\textwidth}
{\scriptsize
  \begin{lstlisting}
  void F(int &a) {
    a = 1; /*@\label{line:store-a}@*/
    flush(&a);
  }

  void main() {
    x = 1;
    F(y);
  }
\end{lstlisting}
\caption{A buggy execution\label{fig:buggy-execution}}
  }
\end{subfigure}
\begin{subfigure}{0.4\textwidth}
  {\scriptsize
  \begin{lstlisting}
  void F(int &a) {
    a = 1;
    flush(&a);
  }

  void main() {
    x = 1;
    F(x);
  }
  \end{lstlisting}
  }
\caption{A bug-free execution\label{fig:bug-free-execution}}
\end{subfigure}
\caption{Assume that $\code{x}$ and $\code{y}$ reside on different cache lines and are escaped and clean initially.}
\end{figure}
\fi

\subsection{Handling Arrays\label{sec:array}}
To handle arrays, we abstract
array writes as a pair of an array reference and an index.  Formally,
we model an array element $\aloc = \tuple{\aref, n} \in \ALoc_a$ as a
reference (variable) $\aref \in \Refs$ and a  index $n \in
\mathbb{Z}^{0+}$ from that reference.  We  abstract the array index
using the variable that provided the value for the array dereference
operation.  Thus, our abstraction is only able to track dirty array
elements as long as the original index variable exists.  To ensure
soundness, when a function writes to some PM array, we 
require the written element to be flushed before the function returns or the index variable is changed/lost.

We conservatively assume all array elements have escaped and only
compute a map $\PMap{a} \in \ALoc_a \times \Pstate$ from array
elements to persistency states. The transfer functions for the array persistency analysis can be straightforward obtained
from Figure~\ref{fig:perstransfer}
by replacing the addresses being stored to and loaded from
with the array index variable.
At a function exit, we report warnings if the map $\PMap{a}$
contains any element whose persistency state is not clean.

\subsection{Interprocedural Violation Detection\label{sec:inter-error}}

The robustness violation detection mechanism for interprocedural
analysis is the same as the intraprocedural one, except that
robustness violations that involve multiple functions are also
reported.  Our treatment of pointer arithmetic conservatively reports violations when there are stores to or loads from PM
addresses computed by pointer arithmetic.

\subsection{Relaxing Strict Persistency: An Escape Hatch\label{sec:escape_hatch}}

Robustness violations are not always bugs;  design patterns like
\textit{link-and-persist}~\cite{nvtraverse}, \textit{pointer
  tagging}~\cite{detectable2021li}, and checksums can cause false positives by allowing observable low-level robustness violations without compromising high-level safety.

For example, the RECIPE benchmarks 
sometimes use atomic loads to read data from an atomic variable
that is never mutated.  This is done because the data structure packs
both mutable and immutable state into the same atomic variable.
Flushing such loads can incur high
overheads.

PM programs sometimes
write data and a checksum, and then persist both the data and the
checksum.  When accessing the data, the PM program first verifies the
checksum before using the data.  While this pattern is safe, it can
yield executions that are not equivalent to any execution under strict
consistency and thus violate robustness.

\Tool provides strict persistency by default.  However, if the developer knows that it is safe to escape
from strict persistency, \tool provides annotations
to ignore blocks that are exempt from strict persistency.  When
applied to either a store or an atomic load, this annotation tells
\tool that no flush or fence operation is needed.

\subsection{Termination and Correctness}

All transfer functions in the analysis are monotonic and the lattices are of finite height.  Thus, the dataflow analysis  terminates.  \textbf{The appendix presents correctness proofs.}

\section{Transformation}\label{sec:trans}

This section details \tool's approach to automatically inserting flush
and fence operations to ensure programs are free from missing
flush/fence bugs.

\subsection{Flush Insertion}

When \tool's analysis detects a violation, \tool fixes it by
immediately flushing the PM locations after they become \dirty using
\code{clwb}, which changes their state to \clwb in our analysis. Our
experience shows that conservatively flushing on atomic loads
can incur a significant overhead.
To alleviate this problem, we implement the FliT
transformation~\cite{flit} for atomic memory accesses, which uses
counters to signal whether all stores to a PM location have been
flushed.  By reading this counter after an atomic load, we
can determine whether the store may not have
completed its flush and thus the load must help.  This improves the
performance of many atomic-load-heavy benchmarks, as can be seen in
the evaluation results of Section~\ref{sec:eval}.

The flush instruction can be chosen from \code{clflush}, \code{clwb},
and \code{clflushopt}. We choose \code{clwb} as it has weaker ordering
properties than \code{clflush} and may retain the cache line on some
architectures as compared to \code{clflushopt}, potentially leading to
better performance.

If the length of dirty bytes in the object is not constant, which
could occur when modifying an object using functions like as
\code{memcpy}, a call to a variable-sized flush function is inserted
that issues \code{clwb} instructions over the entire range.

\subsection{Fence Insertion}
After flushes are inserted, \tool uses the detected robustness violations to perform fence insertions. We
insert a \code{sfence} before the instruction where a violation is
reported, eliminating the violation.

Specifically, fences are inserted for the two types of \tool
robustness violations as follows:
\begin{enumerate}
    \item \label{error1} A fence is inserted at function exit when
      there are objects in the $\clwb$ state.
    \item \label{error2} When there are two escaped and non-clean
      objects at the same time, a fence is inserted right before the
      second object becomes escaped and non-clean.
\end{enumerate}

Figure~\ref{fig:fence} provides an example of the second case,
showing a block of code with PM-pointers \code{x} and \code{y},
which are escaped and clean at the beginning.
In the unmodified version, the address of \code{x} is written
first, making it dirty, and the address of \code{y} is written later, causing a violation. To fix the violation, \tool inserts
a \code{clwb} instruction on \code{x} after the write, making its
state $\clwb$, and before the address of \code{y} is written, it
inserts a fence, removing the
violation.

\begin{figure}[!h]
\begin{subfigure}{.49\linewidth}
{\scriptsize
  \begin{lstlisting}
//x: <esc,clean>
*x = 1; 
//x: <esc,dirty>
...
//y: <esc,clean>
*y = 1; 
//x: <esc,dirty>, 
//y: <esc,dirty>
//violation!!
\end{lstlisting}
\caption{Before insertions}
}
\end{subfigure}
\begin{subfigure}{.49\linewidth}
  {\scriptsize
  \begin{lstlisting}
//x: <esc,clean>
*x = 1; 
//x: <esc,dirty>
clwb(x);
//x: <esc,clwb>
...
fence();
//x: <esc,clean>
*y = 1; 
//y: <esc,dirty>
//no violation
  \end{lstlisting}
  }
  \caption{After insertions}
  \end{subfigure}
    \caption{Example Insertion of Flushes/Fences.  esc = escaped}
  \label{fig:fence}
\end{figure}

Cases that \tool does not analyze precisely, such as pointer
arithmetic on PM addresses, involve only one object
and are handled by conservatively inserting a fence. The correctness
of this approach follows from our proof for the error detection
mechanism---the inserted fences turn objects that trigger
violations to \clean before the violation, no violation remains
after the fence insertion procedure completes.

Atomic loads turn object states to 
\dirty in our analysis due to potentially unflushed stores from other threads. However, flushes
of previous writes to PM only need to be serialized before the next
store to PM and not the next atomic load. In other words, there is no
need for fences between multiple atomic loads. Thus, fences
for atomic loads are deferred until the next store to PM
to avoid inserting redundant
fences.

\section{Evaluation \label{sec:eval}}

In this section, we present the evaluation of \tool on several benchmarks in order to answer the following research questions:
\begin{itemize}
    \item RQ1 Applicablility: Does \tool analyze PM programs in a reasonable time?
    \item RQ2 Performance: Do programs transformed by \tool have performance similar to the original programs?
    \item RQ3 Developer Burden: How much developer involvement does \tool require to relax the persistency requirements? 
\end{itemize}
We start by describing our system setup, the settings of \tool, and the benchmarks. Finally, we discuss how our research questions can be answered by our evaluation findings. 

\noindent \textbf{System Setup}: While one  motivation for this
work is CXL shared memory, it is not 
commercially available yet.  Thus, we have evaluated \tool
on Intel Optane PM, which requires the same use of flush and
fence instructions.

\Tool is implemented as a transformation pass in LLVM. All
benchmarks are compiled with clang/clang++ and  LLVM 
with optimization level O3. We used a Ubuntu 22.04.4 machine with a 16
core 2.4 GHz Intel Xeon Silver 4314 processor, 256 GB of RAM, and 256
GB of Optane PM.

\noindent \textbf{\tool Settings:} We evaluate \Tool with three
optimization settings to understand the contribution of
each optimization. 
\textit{PMRobust\textsubscript{base}} uses our alias analysis to insert a flush and a
fence immediately after every store and atomic load to PM.  \textit{PMRobust\textsubscript{opt}} adds our escape and
persistency state analysis to delay insertion of fences, as described
in Section~\ref{sec:trans}. 
\textit{PMRobust\textsubscript{flit}} adds  FliT on
top of \textit{PMRobust\textsubscript{opt}}.

\subsection{Benchmarks}

Our evaluation includes RECIPE~\cite{recipe}, a set of
high-performance concurrent index data structures modified to be crash-consistent for persistent
memory. RECIPE uses PMDK~\cite{pmdk}'s libvmmalloc to convert all
dynamic memory allocations to PM allocations. Of the data structures
in RECIPE, P-HOT, WOART, and LevelHash fail to compile with LLVM,
and we evaluate the remaining six (P-ART, P-BwTree, P-CLHT,
P-Masstree, FAST\&FAIR, CCEH) using the YCSB
benchmark~\cite{ycsb}. During evaluation, we found timing-related
bugs in CCEH and FAST\&FAIR and
discard those runs when they occur. These bugs also occur in the
unmodified programs. By comparison with the transformed version, we also discovered a number of
missing flushes and fences in the functions \code{getChild()} and
\code{getChildren()} in P-ART's node classes, which allow the client to retrieve values that are not yet persisted in the index. If the client then writes to PM
locations based on the retrieved value and a crash happens, the later update could be persisted while its source in the index might not, \ie a crash consistency bug. We added these flushes and
fences to P-ART. 

We also include Memcached~\cite{memcached}, a popular distributed
memory object caching system. Specifically, we choose a version of Memcached
that supports persistent memory via PMDK's libpmem. We use the
memaslap load testing tool from libMemcached~\cite{libmemcached} for
benchmarking.

Lastly, we evaluate our tool on PMDK~\cite{pmdk}, the most widely-used
open-source libraries for programming persistent memory, using a data
store implementation which is provided as one of its example programs. The data
store has a swappable backend that allows choosing from among seven
map data structures: \code{btree}, \code{rbtree}, \code{ctree},
\code{hashmap\_atomic}, \code{hashmap\_tx}, \code{hashmap\_rp}, and
\code{skiplist}.

For each benchmark, we show the average throughput of the original and
the transformed program over 5 runs for RECIPE and over 10 runs for 
memcached and PMDK's data store, and display the standard deviation as
error bars. In the transformed programs, we first removed the existing
flushes and fences from the original programs and then our tool
transforms them by inserting flushes and fences.

\subsection{RQ1: Applicability}
To answer RQ1, we measure the execution time of \tool pass and under the \textit{PMRobust\textsubscript{opt}} setting for each of the benchmarks. We then compare them against the execution time of the PM bug repair tool PMBugAssist~\cite{pmbugassist}, which uses SMT solving to generate targeted fixes for PM bugs contained in the execution trace, but is not able to guarantee the repaired program is free of further PM bugs. Although PMBugAssist belongs to a separate paradigm, there is currently no other PM bug repair tool that takes a purely static approach like \tool does. We use the comparison to demonstrate that \tool is able to run in reasonable time on the same programs as PMBugAssist, while providing stronger correctness guarantees. Thus, in addition to previously mentioned benchmarks, we include measurements on PMDK test programs evaluated by PMBugAssist. The times are presented in Table~\ref{table:Perf} along with the code size of the benchmarks. The time reported for PMBugAssist is the total time it takes on our machine to fix all bugs provided for each benchmark in the original evaluation.

\begin{table}[h!]
\begin{center}
{\scriptsize
\renewcommand{\arraystretch}{1.1}
\begin{tabular}{ |c|c|c|c|c|c|c|}
 \hline
 Benchmark & PMR Time(s)& PBA Time(s) & Code Size (KLOC)\\
 \hline
   data\_store        & 20.8 & N/A     & 95.2$^*$\\
   Memcached          & 199.1& 1861.49 & 23\\
   RECIPE             & 64   & 5.94    & 40.3\\
   obj\_constructor   & 17.9 & 0.1     & 95.4$^*$\\
   obj\_first\_next   & 17.7 & 0.3     & 95.5$^*$\\
   obj\_mem           & 18.1 & 149.5   & 95.3$^*$\\
   obj\_memops        & 28.3 & 0.3     & 95.8$^*$\\
   obj\_toid          & 18.1 & 0.1     & 95.3$^*$\\
   rpmemd\_db         & 12.2 & 0.1     & 25.3$^*$\\
   pmemspoil          & 14.8 & 0.1     & 47.8$^*$\\
   pmem\_memcpy       & 11.9 & 0.3     & 46.7$^*$\\
   pmem\_memmove      & 11.9 & 0.2     & 46.7$^*$\\
   pmem\_memset       & 11.9 & 198.9   & 46.6$^*$\\
   pmreorder\_simple  & 10.3 & 4810.0  & 46.6$^*$\\
   pmreorder\_flushes & 10.4 & 5911.8  & 46.6$^*$\\
   pmreorder\_stack   & 10.4 & 1.2     & 46.6$^*$\\
 \hline
   \multicolumn{4}{l}{\footnotesize $^*$ including sublibraries of PMDK}
\end{tabular}
}
\iffalse
libpmem 46.5 KLOC
libpmemobj + libpmem 95.2 KLOC
\fi
\end{center}
\caption{Runtime of \tool vs. PMBugAssist. PMR stands for \tool, PBA stands for PMBugAssist. N/A means not included in PBA's original evaluation }\vspace{-.2cm}
\label{table:Perf}
\end{table}

Table~\ref{table:Perf} shows that \tool's execution time roughly correlates with the code size, whereas the execution time of PMBugAssist varies widely, ranging from less than a second to close to an hour, depending on the length of execution traces, the number of bugs, and the SMT solving process. These measurements support the claim that \tool is applicable to programs targeted by prior tools, but with more consistent analysis times and stronger guarantees. 

\subsection{RQ2: Performance}\label{ref:rq2}
Here we explore RQ2 by comparing the performance of \tool's transformed programs with the originals. 

\noindent \textbf{RECIPE:}
We follow RECIPE's evaluation procedure~\cite{recipe} by
testing the ordered indexes (P-ART, P-BwTree, P-Masstree, FAST\&FAIR)
on two types of keys---\textbf{randint} (8 byte random integer keys)
and \textbf{string} (24 byte YCSB string keys), all uniformly
distributed, with YCSB workloads A, B, C, and E, and the unordered
index P-CLHT and CCEH on workloads A, B, and C with uniformly
distributed randint keys. In both cases, we set the thread count to 16
and first populate the index with 64M keys using LoadA, and then run
the respective workloads that insert and/or read the keys. 
%Note that fixing the crash consistency bug we discovered caused a decrease in the original P-ART's performance on read-heavy workloads, especially workload E that performs range scans.

We discovered a bottleneck when running YCSB workload E with the transformed P-BwTree that heavily uses an iterator data structure to traverse retrieved values. An iterator and the data it contains is
never reused after a crash, yet its data is allocated in PM. This
unnecessary usage of PM occurs because the RECIPE benchmarks, as
research prototypes, do not use separate memory allocators for PM and
DRAM allocations, but simply use libvmmalloc to perform all dynamic
allocations using PM. In a proper implementation, iterator data
should be allocated in DRAM. Thus, we eliminated the flush insertion on
iterator contexts with 5 annotation pairs, where each 
annotation pair ignores exactly only one program statement. Figure~\ref{fig:recipe_perf} shows the resulting
throughputs.

\begin{figure}[!htb]
\centering
\begin{subfigure}{.45\linewidth}
\includegraphics[width=\linewidth]{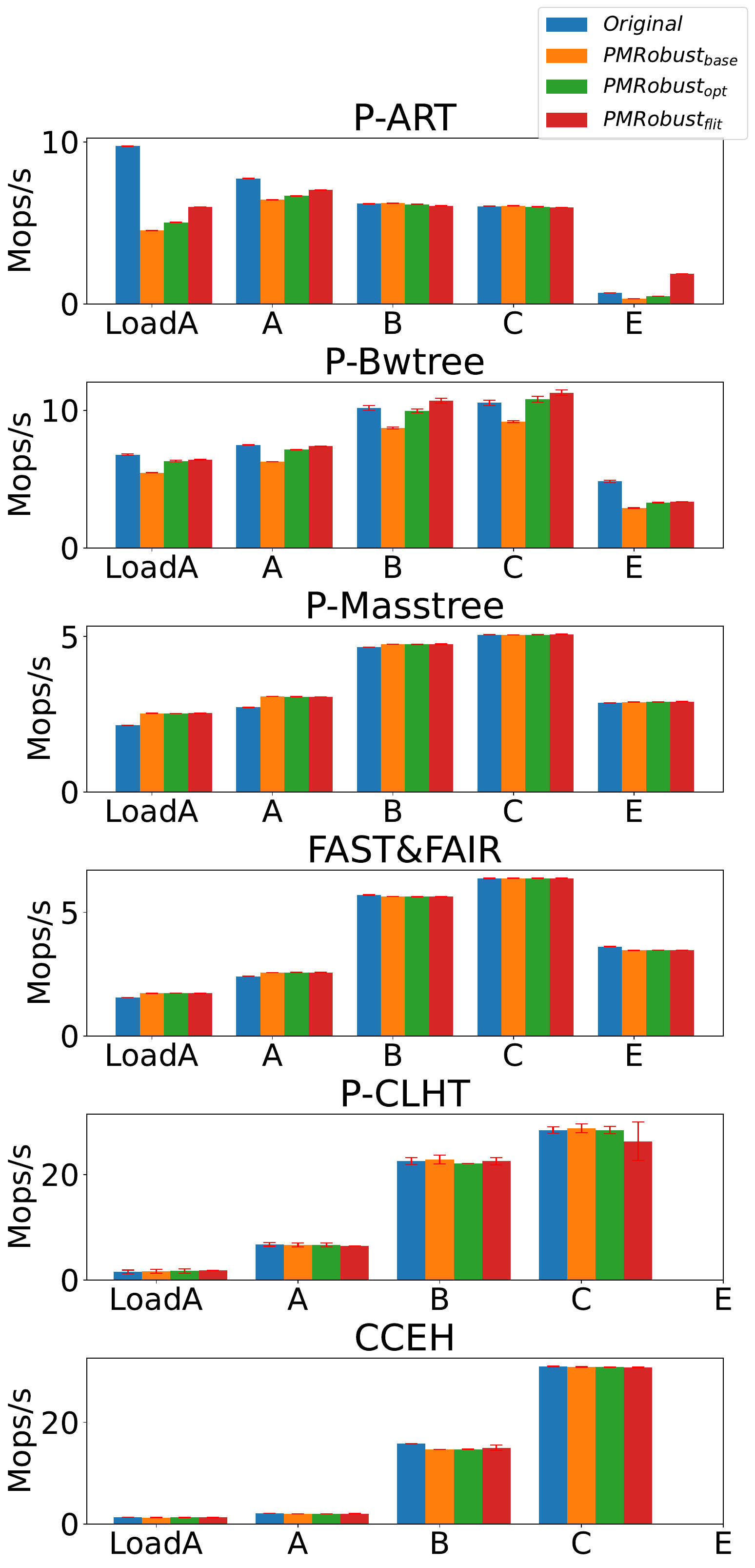}
\caption{Integer keys}
\end{subfigure}
~
\begin{subfigure}{.45\linewidth}
\includegraphics[width=\linewidth]{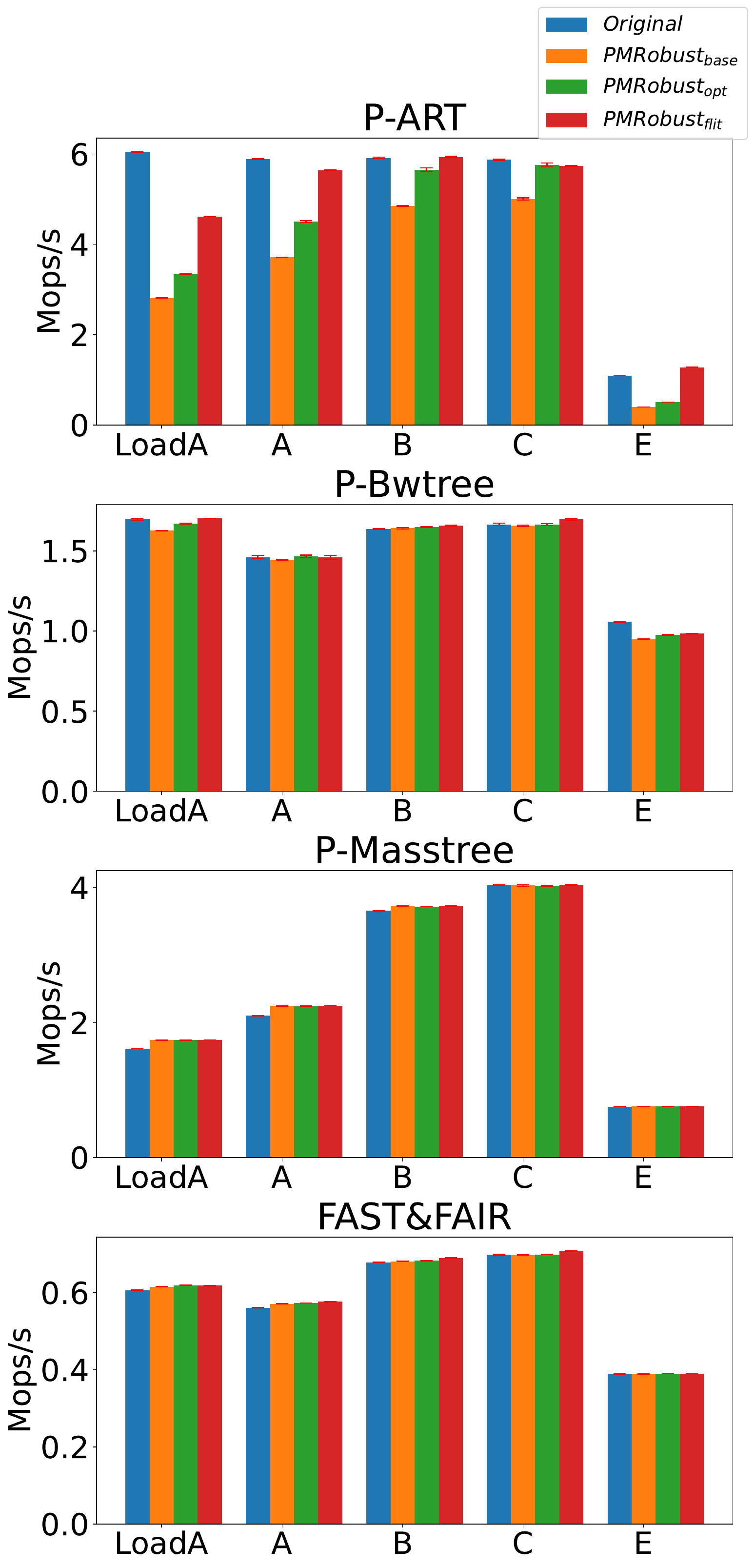}
\caption{String keys}
\end{subfigure}
\caption{Throughput of RECIPE indexes.  Larger is better.}
\label{fig:recipe_perf}
\end{figure}

\noindent\textbf{Memcached:}
We evaluate Memcached's throughput with 16 threads and 100K
operations using memaslap. The workload uses 16-byte keys and
1024-bytes values. The proportion of get and set operations is 90/10.
Figure~\ref{fig:memcached_perf} reports the throughputs.

\begin{figure}[!htb]
\vspace{-.3cm}
\begin{center}
\includegraphics[width=0.5\linewidth]{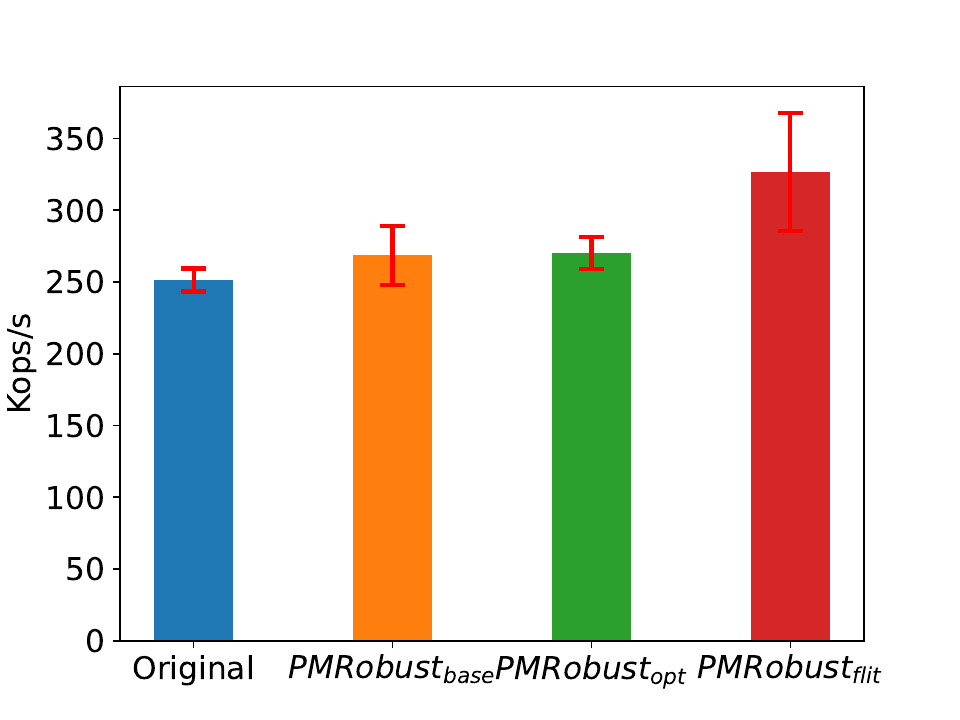}
\end{center}
\vspace{-.3cm}
\caption{Throughput of Memcached.  Larger is better.}
\label{fig:memcached_perf}
\end{figure}

\noindent\textbf{PMDK:}
During the data store benchmark, 100K randomly generated 8-byte
integer keys are inserted into the data store and then
removed. Figure~\ref{fig:data_store_perf} presents the results of the
seven map data structures.

\begin{figure}[!htb]
\vspace{-.4cm}
\begin{center}
\includegraphics[width=.8\linewidth]{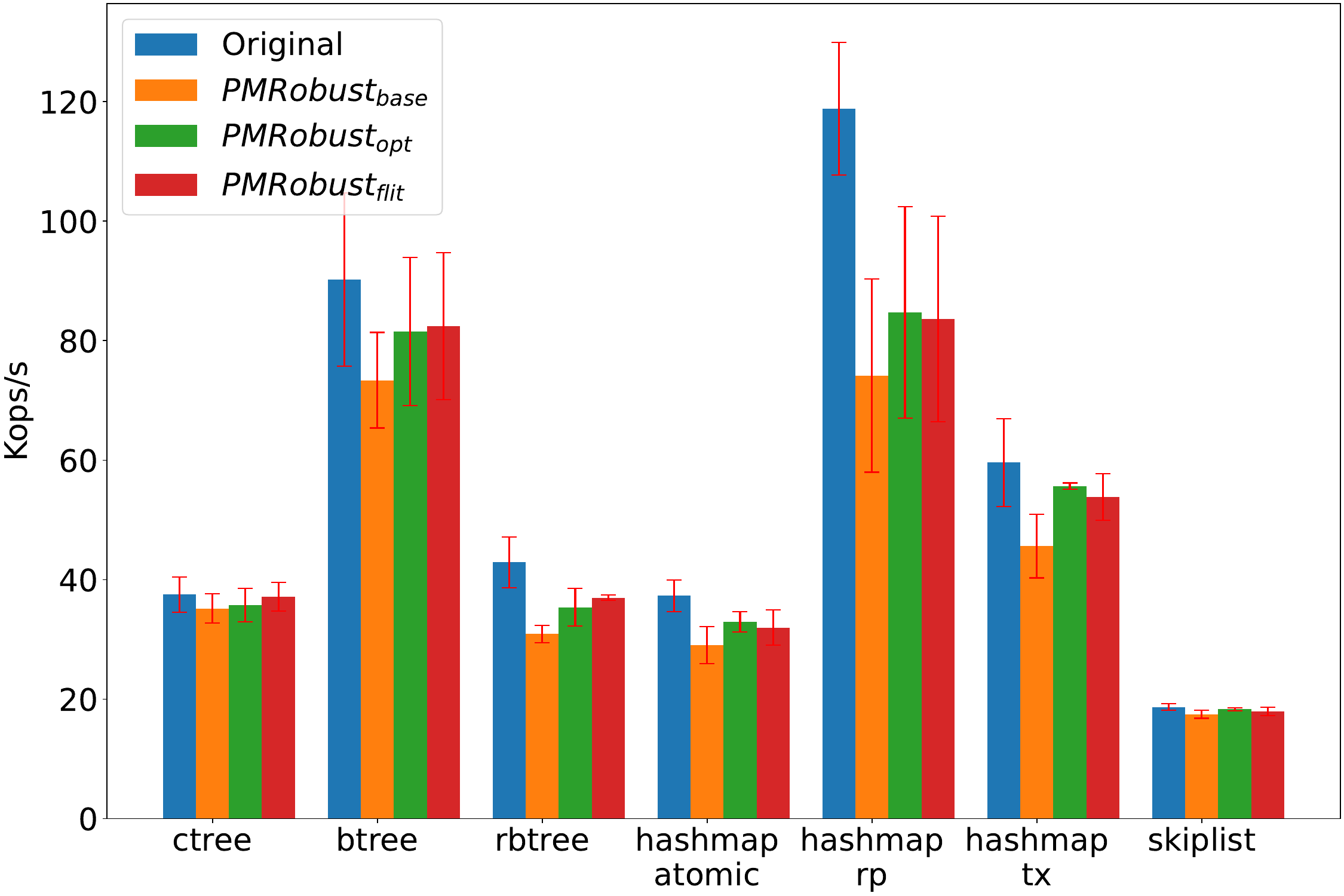}
\end{center}
\vspace{-.2cm}
\caption{Throughput of PMDK's data store.  Larger is better.\label{fig:data_store_perf}}
\end{figure}

\vspace{.2cm}
First, we see from the evaluation results that the base
transformation achieves performance on-par with the original programs
on a number of benchmarks, including P-CLHT, P-Masstree, Fast\&Fair,
CCEH, and Memcached. This suggests that in these benchmarks, extra
flushes and fences inserted by the base approach do not appear
on performance-critical paths. For the other benchmarks that have a
noticeable gap between the performance of
\textit{PMRobust\textsubscript{opt}} and the original program,
\textit{PMRobust\textsubscript{opt}} is able to either reduce, or in
some cases completely eliminate the performance gaps. The FliT
transformation of \textit{PMRobust\textsubscript{flit}} further
reduces the gap on many workloads for P-ART and P-BwTree due to their
frequent atomic loads. The improvement is most notable for P-ART,
which performs many atomic loads during its tree traversal on each key
insertion. \textit{PMRobust\textsubscript{flit}} outperforms the
original program significantly on workload E as it eliminates most of
the overhead from the originally missing flushes mentioned in
Section~\ref{ref:rq2}.  It also outperforms the original
Memcached by approximately 20\%, which can be explained by the fact
that \tool inlines flush instructions whereas the original benchmark
called a flush function.

Across all benchmarks, the geometric mean overhead over the original
programs is 11.21\% for \textit{PMRobust\textsubscript{base}}, and
6.41\% for \textit{PMRobust\textsubscript{opt}}, and only \overhead
for \textit{PMRobust\textsubscript{flit}}.

Overall, the performance results show that \tool's automatic flush and
fences insertion is able to match the performance of the original
programs on most benchmarks using our dataflow analysis and FliT, while only using a few user annotations. This answers the research question in the affirmative---\tool is able to produce programs with performance close to the originals,

\subsection{RQ3: Developer Burden}\label{ref:rq3}
In this section we answer RQ3 regarding the extent users of \tool need to be involved to relax the persistency requirement using manual annotations. As we noted in \ref{ref:rq2}, there was only one instance of significant bottleneck that we removed with an annotation during the performance evaluation. We identified the source of the bottleneck by profiling. 

Recall that we noted that this one example was unusual, and due to the fact that RECIPE does not implement a reasonable memory allocation strategy.  We would not expect to need annotations for this same reason in non-research software.
Based on our experience, \tool does not impose too much burden on developers as bottlenecks should be rare, and fixing them only requires familiarity with profilers, which would be reasonable to expect from developers working on low-level software such as PM programs.

\subsection{Threats to Validity}

\Tool currently does not implement support for function pointers and may produce imprecise results for them.  This can be addressed by ensuring that all objects are clean before making a call using a function pointer and ensuring that if a function has its address taken, that all objects must be clean at exit.  This could also potentially be handled by pointer analysis.

Our current implementation uses whole-program analysis.  As a result, we perform
our analysis and transformation passes after all translation units are
linked together.
This
requirement can be removed by treating interactions with code outside of
the current compilation unit conservatively. 

\section{Related Work}

There is work on checking/testing PM programs to find bugs. In
particular, XFDetector~\cite{xfdetector} uses a finite state machine
to track the consistency and persistency of persistent
data. PMTest~\cite{pmtest} lets developers annotate a program with
checking rules to infer the persistency status of writes and ordering
constraints between writes. Pmemcheck~\cite{pmemcheck} checks how many
stores were not made persistent and detects memory overwrites using
binary rewriting. Yat~\cite{yat} model checks
PM programs. 
Agamotto~\cite{agamotto} finds bugs in PM programs by
using symbolic execution. An algorithm by Huang et al.~\cite{pminvariant} infers invariants from PM programs that are then used to check for bugs.
Although these tools are able to find many bugs, none of these tools
can assure the absence of flush/fence bugs like \tool can.  POG~\cite{pog} and
Pierogi~\cite{pierogi} provide logics that can be used to manually
reason about program behaviors.

A line of work~\cite{nv-htm,crafty-pldi20,persistent-htm-giles-2017,dudetm} uses (software or hardware) transactions to provide
(failure and thread) atomicity. NVL-C provides language and compiler
support for the use of transactions to access non-volatile
memory~\cite{nvlc}.  While NVL-C can provide crash consistency, it
does so by incurring the overheads of using transactions to provide
crash consistency.  Another line of
work~\cite{atlas-follow-up,atlas,nvthreads,justdo-logging,ido}
advocates use of locks or synchronization-free
regions~\cite{persistency-sfr}.  Memento~\cite{memento} provides
detectable checkpointing---it extends standard checkpoint with support
to allow the system to be able to detect the status of in flight
operations when the crash occurred.  These approaches typically incur
large overheads to support the necessary logging.

StaticPersist~\cite{staticpersist} is a static analysis to
determine which objects must be allocated in PM.  The
idea is to use annotations to declare a set of durable roots, and the
analysis determines which objects are reachable from the 
roots. AutoPersist~\cite{autopersist} is a Java extension to support
NVM.  Developers specify durable roots and when an
object becomes reachable from a durable root, AutoPersist moves it to
PM.  Both
StaticPersist and AutoPersist provide a higher-level API for
programming PM, while \tool targets a
lower-level model by automatically inserting flush and
fence operations.

Hippocrates~\cite{hippocrates-asplos21} and
PMBugAssist~\cite{pmbugassist} insert flushes
and fences to repair PM bugs.  Hippocrates chooses the placement of
flushes and fences using a reduction procedure, whereas PMBugAssist
uses a SMT solver. They both focus on
repairing specific PM bugs given as program input and rely on other PM
bug detection tools to produce bug traces, which is different from our
task of exhaustively detecting potential PM bugs and fixing them at
the same time.

Our work is related to a line of work on static fence
insertion for concurrent programs to ensure sequential consistency~\cite{musketeer2017, lee2000relaxed, fang2003automatic}. The static approach to fence insertion has so far not been applied to PM programs, a
gap which we bridge. In terms of techniques, previous work such as musketeer~\cite{musketeer2017} and pensieve~\cite{fang2003automatic} focused on variants of delay set analysis~\cite{delayset} combined with flow-insensitive escape analysis to minimise fence insertions, whereas we make use of a flow-sensitive escape analysis to avoid redundant fence insertions. 

FliT~\cite{flit} presents a technique that uses counters to
eliminate flush operations on atomic loads, which we applied in \tool
to further reduce the overhead it introduces.

\section{Conclusion\label{sec:conc}}

Correctly using flush and fence operations is notoriously hard.   \Tool employs  compiler analysis and transformation to automatically insert flush and fence instructions, providing strong persistency via robustness, a sufficient condition for the correct usage of flush and fence operations in PM programs.  \Tool ensures the absence of flush and fence bugs and eliminates the time-consuming and error-prone task of manually inserting these operations.

\section*{Acknowledgements} We would like to thank the anonymous reviewers for their thorough and insightful comments that helped us improve the paper.  This work is supported by National Science Foundation grants CCF-2006948, CCF-2102940, and CCF-2220410. 

\bibliographystyle{plain}
\bibliography{paper}

\clearpage
\appendix

\setcounter{theorem}{0}
\section{Proof}

\begin{theorem}\label{thm:thm1}
Supposed that function $F$ calls $G(x_1, ..., x_n)$ and some objects are
escaped and non-clean in the program state of $F$ right before calling
$G$.  If the \textit{marksObjDirEsc} bit is set in $G$ under the
calling context, while none of $x_1, ..., x_n$ is escaped and non-clean,
then this is a robustness violation.
\end{theorem}

\begin{proof}
If the \textit{marksObjDirEsc} bit is set in $G$ under the calling
context, and none of $G$'s parameters is escaped and dirty, then $G$
must make some object $O$ escaped and dirty, and $O$ is different from
the objects that are already escaped and dirty before calling $G$ in
$F$.  Therefore, calling $G$ causes a robustness violation.
\end{proof}

\begin{theorem}\label{thm:thm2}
Suppose that function $F$ calls $G(x_1, ..., x_n)$ and no objects are
escaped and non-clean in the program state of $F$ right before calling
$G$. Then if calling $G$ causes any robustness violation, the
violation will be detected while analyzing $G$ with the calling
context.
\end{theorem}

\begin{proof}
Since no objects are escaped and dirty in the program state of $F$
before calling $G$, the calling context of $G$ has all of its elements
being $\tuple{\text{captured}, \text{clean}}$.
So if calling $G$ causes any robustness
violation in $F$, the violation will be detected while analyzing $G$
with the calling context.
\end{proof}

\begin{theorem}\label{thm:thm3}
Supposed that function $F$ calls $G(x_1, ..., x_n)$ and some objects are
escaped and non-clean in the program state of $F$ right before calling
$G$.  Suppose some of $x_1, ..., x_n$ is escaped and non-clean.
Then if there is a robustness violation caused by calling $G$ from $F$,
some robustness violation is reported.
\end{theorem}

\begin{proof}
Suppose that there is one escaped and dirty object $O$ in the program
state of $F$ before calling $G$.  Then $O$ must be one of $G$'s
parameters.  Therefore, the only case that causes a robustness
violation is where $G$ makes some other object escaped and dirty
before flushing $O$.  This violation can be detected when analyzing
$g$ with its calling context.

Now suppose that there are more than one escaped and dirty objects in
the program state of $F$ before calling $G$.  Without loss of
generality, assume that there are exactly two such objects $O_1$ and $O_2$.
Suppose $O_1$ is passed into $G$ while $O_2$ is not.  If $G$ makes any
object other than $O_1$ escaped and dirty, then this is the same case
as the previous paragraph.  If $G$ stores to $O_1$, there is a
robustness violation between the pair $O_1$ and $O_2$. However, a
robustness violation already exists between the pair before calling
$G$.
\end{proof}

\begin{lemma}\label{lemma:esc}
If a PM location \code{x} is reachable from persistent data structure roots,
then the escape analysis will mark \code{x} and its alias as escaped.
\end{lemma}

\begin{proof}
We will prove the statement for intraprocedural analysis first,
and then prove for the interprocedural analysis.

\textit{Case 1:}
If \code{x} is stored to some data structure via \code{*y = x},
then \code{x} and all elements in its alias set $\AliasMap{(\code{x})}$
are marked as escaped by the second transfer function
in Figure~\ref{fig:esctransfer}.
Furthermore, if \code{x} is later loaded from \code{y}, say \code{a = *y},
then \code{a} is an alias of \code{x}.
Note that although \code{a} is not in the alias set $\AliasMap{(\code{x})}$
of \code{x}, \code{a} is also marked as escaped by
the third transfer function in Figure~\ref{fig:esctransfer}.

\textit{Case 2:}
If \code{x} is added to the alias set $\AliasMap{(\code{a})}$
of some variable \code{a} via \code{a = x}, and
later \code{a} is made escaped via \code{*y = a},
then the second transfer function in Figure~\ref{fig:esctransfer}
marks the entire set $\AliasMap{\code{a}}$ as escaped.

\textit{Case 3:}
If \code{x} has already escaped, and later an alias of \code{x}
is created via \code{y = x}, then the first transfer function
in Figure~\ref{fig:esctransfer} marks \code{y} as escaped.

For interprocedural analysis, we only need to consider two cases.

\textit{Case 1:}
Supposed that \code{x} has already escaped in some function \code{H},
then \code{y} becomes an alias of \code{x} in some function
\code{G(x, y, ...)}, where \code{H} is the caller of \code{G}.
Note that the function cached results also contain the may-alias
information between function parameters and the return value.
So the analysis will use the cached result to mark \code{y} as escaped.

\textit{Case 2:}
Supposed that \code{y} becomes an alias of \code{x} in some function
\code{F(x, y, ...)}, and then \code{x} escapes in some function
\code{G(x, ...)} (or \code{y} escapes in \code{G(y, ...)}),
where \code{F} and \code{G} share the same caller \code{H}.
Right after calling \code{F} in \code{H}, \code{H} has the information
that \code{y} may alias \code{x}.
If \code{x} escapes in \code{G(x, ...)},
then when the analysis uses the cached result of \code{G}
to mark \code{x} as escaped, it also marks the variables that may alias
\code{x} as escaped.
The case is similar when we have \code{y} escape in \code{G(y, ...)}
instead of \code{x} escaping in \code{G(x, ...)}.
\end{proof}

\begin{figure}[!h]
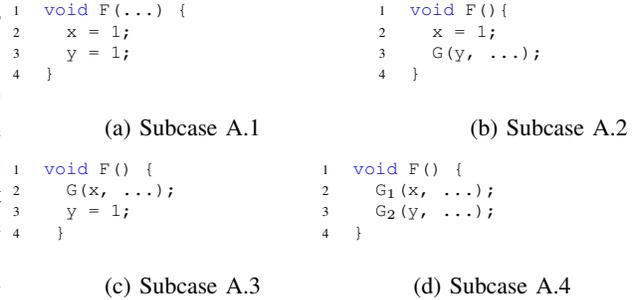

\begin{subfigure}{0.22\textwidth}
{\footnotesize
  \begin{lstlisting}
 void F(...) {
   x = 1;
   y = 1;
 }
\end{lstlisting}
\caption{Subcase A.1}
  }
\end{subfigure}
\begin{subfigure}{0.22\textwidth}
  {\footnotesize
  \begin{lstlisting}
 void F(){
   x = 1;
   G(y, ...);
 }
  \end{lstlisting}
  }
\caption{Subcase A.2}
\end{subfigure}
\begin{subfigure}{0.22\textwidth}
  {\footnotesize
  \begin{lstlisting}
 void F() {
   G(x, ...);
   y = 1;
  }
  \end{lstlisting}
  }
\caption{Subcase A.3}
\end{subfigure}
\begin{subfigure}{0.22\textwidth}
  {\footnotesize
  \begin{lstlisting}
 void F() {
   G$_1$(x, ...);
   G$_2$(y, ...);
 }
  \end{lstlisting}
  }
\caption{Subcase A.4}
\end{subfigure}

\caption{Assume that $\code{x}$ and $\code{y}$ are PM locations that reside on different cache lines and are escaped and clean initially.\label{fig:proof-example}}
\end{figure}

\begin{figure}[!htbp]
\vspace{-.45cm}
\begin{subfigure}{.15\textwidth}
\vspace{.45cm}
\begin{center}
\includegraphics[scale=0.50]{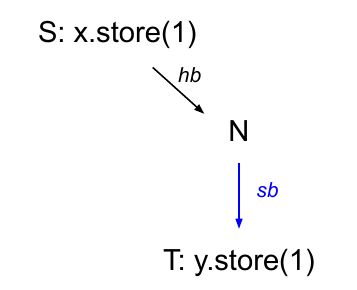}
\end{center}
\vspace{.45cm}
\caption{Subcase B.1}
\end{subfigure}
\begin{subfigure}{.15\textwidth}
\vspace{.45cm}
\begin{center}
\includegraphics[scale=0.50]{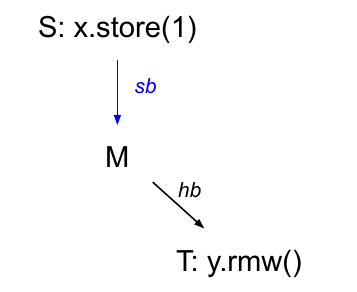}
\end{center}
\vspace{.45cm}
\caption{Subcase B.2}
\end{subfigure}
\begin{subfigure}{.15\textwidth}
\begin{center}
\includegraphics[scale=0.50]{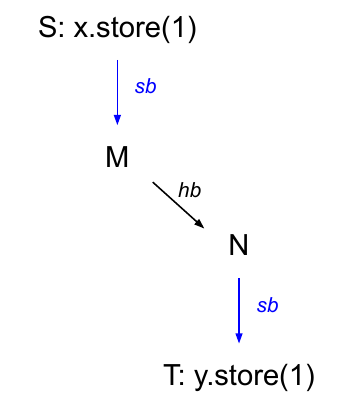}
\end{center}
\caption{Subcase B.3}
\end{subfigure}
\caption{Assume that $\code{x}$ and $\code{y}$ are PM locations that reside on different cache lines and are escaped and clean initially, where \textit{sb} represents the sequenced-before relation, and \textit{hb} represents the happens-before relation.\label{fig:proof-inter}}
\end{figure}

\begin{theorem}\label{thm:report}
If a program has a robustness violation, then it will be detected by \tool.
\end{theorem}

\begin{proof}
A robustness violation must involve two stores $S$ and $T$
that write to PM locations \code{x} and \code{y} on
different cache lines such that \code{x} and \code{y}
are reachable from persistent data structure roots,
that $S$ happens before $T$,
and that \code{x} is not flushed before the store to \code{y}.
By Lemma~\ref{lemma:esc}, we can assume that
both \code{x} and \code{y} have been marked as escaped.
We next enumerate the cases and show that our analysis reports some error
or warning. For simplicity, we assume \code{x} and \code{y}
are two different PM objects.
The same proof will apply if they are fields of
the same PM object that reside on different cache lines.

\paragraph{Case A} We first consider the case where the stores $S$
and $T$ are in the same thread.

Figure~\ref{fig:proof-example} presents a few subcases.
Without loss of generality, we assume that functions \code{G, G$_1$, G$_2$}
write to their first parameters without flushing them.
Although there can be chains of function calls in practice, those
cases are not different from the cases presented in
Figure~\ref{fig:proof-example}.
\textit{Subcase A.1} is obvious, as the analysis
detects two escaped and dirty objects at \code{y = 1} and reports an error.
For \textit{Subcase A.2}, the analysis also reports some robustness violation
according to Theorem~\ref{thm:thm3}.
\textit{Subcase A.3} is similar to \textit{Subcase A.1},
as we assume \code{G} writes to its first parameter.
\textit{Subcase A.4} is similar to \textit{Subcase A.2}.

\paragraph{Case B} Now we consider the case where the stores $S$
and $T$ are in two different thread.
Since we assume $S$ happens before $T$, there are three possible subcases
as presented in Figure~\ref{fig:proof-inter}.
The $\stackrel{\textit{hb}}{\rightarrow}$ edges represent the happens-before
relation, and the $\stackrel{\textit{sb}}{\rightarrow}$ edges represent the
sequenced-before relation or the program order.

In \textit{Subcase B.1}, since $N$ is sequenced before $T$,
we can assume that operations $N$ and $T$ are in the same function, say $F$.
Since $N$ happens after $S$, $S$ is an atomic store,
and $N$ is either an atomic load or atomic RMW that reads from \code{x}.
In either case, by the transfer functions in Figure~\ref{fig:perstransfer},
\code{x} is dirty right after the operation $N$.
Thus, there are two escaped and dirty objects after the store $T$,
and the analysis reports a bug while analyzing the function $F$.

In \textit{Subcase B.2}, since $M$ happens before $T$ that performs a store
to \code{y}, $T$ is an atomic RMW operation,
and $M$ can be an atomic store or atomic RMW that stores to \code{y}.
Since $S$ is sequenced before $M$, we can assume that $S$ and $M$
are in the same function, say $F$.
By the same reasoning as \textit{Subcase B.1},
there are two escaped and dirty objects after the operation $M$,
and the analysis reports a bug while analyzing the function $F$.

In \textit{Subcase B.3}, the operations $M$ and $N$ establish a
happens-before relation. Although, the operations $S$ and $T$ are
not necessarily atomics, that does not change the proof.
There are three possibilities.
Subcase B.3.1) If $M$ and $N$ are atomic operations that
read from or write to persistent memory locations, then it is similar
to \textit{Subcase B.1} and \textit{Subcase B.2}.
Subcase B.3.2) If $M$ and $N$ are atomic operations on non-persistent memory locations, then $M$ must be a release operation and we require that (a) all cache lines must be clean before doing any release operation and
(b) any function call in which all parameters are clean in the calling context
that contains a release operation marks itself as having done a
release operation, and thus all cache lines must be clean
before making such a function call.
Subcase B.3.3) If $M$ and $N$ are locking operations where $M$ is an release and $N$ is an acquire, it is the same as Subcase B.3.2.

\paragraph{Other Cases}
We will discuss some other cases in this paragraph.
Note that the above proof also applies if \code{x} and \code{y}
are array elements. If \code{x} is some object reachable from
some function parameter, recall
our strategy where we create a new label when \code{x} is dereferenced.
The new label is marked as escaped by the transfer functions in
Figure~\ref{fig:esctransfer}.  Thus, the proof above also applies.
Lastly, if \code{x} is an address computed by pointer arithmetic,
then we will throw a warning right after the store to \code{x}.
\end{proof}

\begin{theorem}
If \tool does not report a robustness violation, then any execution of the program is one under strict persistency.
\end{theorem}

\begin{proof}If \tool does not report a robustness violation, there is at most one escaped and non-clean object at each program point of the data flow analysis, whose transition to non-clean state is not solely due to atomic loads. Then for any execution of the program, such an object is always flushed to the clean state before another one appears, which matches the definition of strict persistency that “persistency memory order is identical to volatile memory order.”
\end{proof}

\begin{theorem}\label{thm:noreport}
After \tool transforms a program by inserting flushes and fences, \tool will not report any robustness violations in the program. 
\end{theorem}

\begin{proof}
Since in the transformation, \tool inserts \code{clwb} after stores whenever it detects a robustness violation, we can assume without loss of generality that all stores are immediately flushed with \code{clwb} instructions. 
Then we need to prove that the inserted fences effectively eliminate all robustness violations. We will proceed by analyzing the cases discussed in Theorem~\ref{thm:report}. 

\paragraph{Case A} We first consider the case where the stores $S$
and $T$ are in the same thread.
For \textit{Subcase A.1} and \textit{Subcase A.2}, fences will be inserted after \code{x = 1}.
For \textit{Subcase A.3} and \textit{Subcase A.4}, fences will be inserted at the exits of functions \code{G} and \code{G$_1$}.
In any case, \code{x} becomes clean before the store to \code{y}. So there is no robustness violation.

\paragraph{Case B} Now we consider the case where the stores $S$
and $T$ are in two different thread.

In \textit{Subcase B.1}, $N$ and $T$ are in the same thread, and $N$ is either an atomic load or atomic RMW that reads from \code{x}. In either case, fences will be inserted 
after $N$ and before $T$, transferring \code{x} to the clean state before the store to \code{y}.

In \textit{Subcase B.2}, $S$ and $M$ are in the same thread, and $M$ is either an atomic store or atomic RMW that stores to \code{y}. Fences will be inserted after $S$ and before $M$, transferring \code{x} to the clean state before the store to \code{y}.

In \textit{Subcase B.3}, there are three possibilities. Subcase B.3.1 is similar to \textit{Subcase B.1} and \textit{Subcase B.2}, so we omit the discussions here. Subcase B.3.2) $M$ is a release operation. Fences will be inserted before $M$ since \code{x} is in the state of \code{clwb}. Thus, the robustness violation is eliminated. Subcase B.3.3 is similar to Subcase B.3.2. 

\textit{Other Cases}. If \code{x} and \code{y} are array elements, fences will be inserted after the store to \code{x} as well as at function exits to eliminate robustness violations. The case is similar if \code{x} is some object reachable from function parameters.
If \code{x} is an address computed by pointer arithmetic, then fences will be conservatively inserted after the stores to \code{x}. 
\end{proof}

\end{document}